\documentclass[11pt,reqno]{amsart}
\usepackage{amssymb}

\usepackage{amsthm,amsfonts,amssymb,euscript,hyperref,color}
\usepackage{comment}
\usepackage{mathtools}
\usepackage{stmaryrd}
\usepackage{mathrsfs}
\usepackage{pifont}
\usepackage{amsfonts,amsmath,amssymb,amsthm,amscd,cancel}
\usepackage{graphicx}
\usepackage[]{geometry}
\usepackage{verbatim}
\usepackage{mathrsfs}
\usepackage{fancyhdr}
\usepackage{footnpag,footmisc}

\newtheorem{theorem}{Theorem}
\newtheorem{lemma}{Lemma}
\newtheorem{proposition}{Proposition}

\newtheorem{remark}{Remark}

\allowdisplaybreaks[4]

\DeclareMathAlphabet{\mathsfsl}{OT1}{cmss}{m}{sl}
\numberwithin{equation}{section}

\newcommand{\D}{\mathrm{d}}

\newcommand{\tr}{\mathrm{tr}}

\renewcommand{\L}{\mathbb{L}}
\renewcommand{\H}{\mathbb{H}}

\def\alphab{\underline{\alpha}}
\def\betab{\underline{\beta}}
\def\chib{\underline{\chi}}
\def\chibh{\widehat{\underline{\chi}}}
\def\chih{\widehat{\chi}}
\def\etab{\underline{\eta}}

\def\Lb{\underline{L}}
\def\mub{\underline{\mu}}

\def\tr{\mathrm{tr}}
\def\omegab{\underline{\omega}}

\def\tensor{\widehat{\otimes}}

\def\ub{\underline{u}}
\def\Cb{\underline{C}}

\def\Lb{\underline{L}}

\newcommand{\Db}{\underline{D}}
\newcommand{\Dh}{\widehat{D}}
\newcommand{\Dbh}{\widehat{\underline{D}}}

\def\nablas{\mbox{$\nabla \mkern -13mu /$ }}
\def\Deltas{\mbox{$\Delta \mkern -13mu /$ }}

\def\divs{\mbox{$\mathrm{div} \mkern -13mu /$ }}
\def\curls{\mbox{$\mathrm{curl} \mkern -13mu /$ }}
\def\ds{\mbox{$\mathrm{d} \mkern -9mu /$}}
\def\gs{\mbox{$g \mkern -9mu /$}}
\def\epsilons{\mbox{$\epsilon \mkern -9mu /$}}

\begin{document}

\title[local existence]{Local existence in retarded time under a weak decay on complete null cones}

\author[Junbin Li]{Junbin Li}
\address{Department of Mathematics, Sun Yat-sen University\\ Guangzhou, China}
\email{lijunbin@mail.sysu.edu.cn}

\author[Xi-Ping Zhu]{Xi-Ping Zhu}
\address{Department of Mathematics, Sun Yat-sen University\\ Guangzhou, China}
\email{stszxp@mail.sysu.edu.cn}

\thanks{The authors are partially supported by NSFC 11271377. The first author is also partially supported by the Fundamental Research Funds for the Central Universities.}

\date{}

\maketitle

\begin{abstract}
In the previous paper \cite{L-Z}, for a characteristic problem with not necessarily small initial data given on a complete null cone decaying like that in the work \cite{Ch-K} of the stability of Minkowski spacetime by Christodoulou and Klainerman, we proved the local existence in retarded time, which means the solution to the vacuum Einstein equations exists in a uniform future neighborhood, while the global existence in retarded time is the weak cosmic censorship conjecture. In this paper, we prove that the local existence in retarded time still holds when the data is assumed to decay slower, like that in Bieri's work \cite{Bie} on the extension to the stability of Minkowski spacetime. Such decay guarantees the existence of the limit of the Hawking mass on the initial null cone, when approaching to infinity, in an optimal way.

\end{abstract}

%\tableofcontents

\setcounter{tocdepth}{1}

\parskip=\baselineskip

\allowdisplaybreaks

\section{Introduction}

In the previous paper \cite{L-Z}, we considered a characteristic problem of the vacuum Einstein equations, where part of the initial data is given on an asymptotically flat  complete null cone. The decay rate given on the complete null cone in \cite{L-Z} inherits that in the remarkable works of the stability of Minkowski spacetime in \cite{Ch-K} and \cite{K-N}, which means that the decay rate we considered, is the same as that of the complete null cones in the solutions in \cite{Ch-K,K-N}. We remark that, the decay rate in the above works does not necessarily satisfy the peeling properties, which are implied by a smooth conformal compactification, but a suitable notion of the future null infinity can still be defined. What we proved in \cite{L-Z} is, that if we start from a complete null cone with the above decay rate, then the solution to the vacuum Einstein equations will always contain a piece of the future null infinity. This gives the local existence in retarded time, and the global existence in retarded time is exactly the weak cosmic censorship conjecture, see the discussions in \cite{Chr90}. It was also proved by Cabet, Chru\'sciel  and Wafo in \cite{C-C-W}, that if the complete null cone attaches the future null infinity smoothly (and therefore the peeling properties are satisfied), then the solution contains a piece of smooth future null infinity.

In \cite{Bie}, Bieri extended the stability of Minkowski spacetime (due to Christodoulou and Klainerman \cite{Ch-K}), under weaker decay and regularity assumptions. Roughly speaking, Christodoulou and Klainerman \cite{Ch-K} considered the asymptotically flat initial data $(\Sigma,\bar{g},\bar{k})$ with the following decay at infinity:
\begin{align*}
\bar{g}_{ij}=\left(1+\frac{2M}{r}\right)\delta_{ij}+o_4(r^{-3/2}),\ \bar{k}_{ij}=o_3(r^{-5/2}),
\end{align*}
where $f=o_m(r^s)$ means $\partial^kf=o(r^{s-k})$\footnote{The exact asymptotic behaviors are written in an integral form.}
 for all $k\le m$.  Bieri considered the following:
\begin{align}\label{Bieri}
\bar{g}_{ij}=\delta_{ij}+o_3(r^{-1/2}),\ \bar{k}_{ij}=o_2(r^{-3/2}).
\end{align}
Bieri's result \cite{Bie} requires one less derivatives and one less power of $r$. One of the common features of both proofs is making use of an optical function $u$, which increases towards future and is normalized at infinity, whose level sets $C_u$ are outgoing null cones extending to infinity. Along these null cones, the geometry tends to being Minkowskian in suitable rate. For example, without regard to the difference in the regularity considered in both works, the curvature components\footnote{See the next section for definitions.}, in the spacetime considered by Christodoulou and Klainerman, behave like:
\begin{align}\label{strong}
\alpha,\beta=o(r^{-7/2}),\ \rho=O(r^{-3}),\ \sigma=O(r^{-3}\tau_-^{-1/2}),\ \betab=O(r^{-2}\tau_-^{-3/2}),\ \alphab=O(r^{-1}\tau_-^{-5/2}),
\end{align}
and behave like the followings in the spacetime considered by Bieri:
\begin{align}\label{weak}
\alpha,\beta,\rho,\sigma=o(r^{-5/2}),\ \betab=O(r^{-2}\tau_-^{-1/2}),\ \alphab=O(r^{-1}\tau_-^{-3/2}).
\end{align}
Here $\tau_-^2=1+u^2$. As a consequence of less decay being assumed on the initial data, the decay of various geometric quantities induced on the null cones $C_u$ becomes slower.

The decay considered by Bieri is expected to be sharp. One reason is that, as pointed out in \cite{Bie}, borderline terms appear in the energy estimates, which suggests that any further relaxation on the decay may cause divergence of the integral and the bootstrap argument cannot be closed. We will make more comments at this point later in a clearer way. Another reason seems to be more reasonable, that is, Bartnik \cite{Bar} proved that, the asymptotically flat initial data verifying \eqref{Bieri} always has a well-defined and unique ADM energy, and it is the optimal decay rate (without regard to the regularity) to guarantee the ADM energy to be uniquely defined (see \cite{D-S}).

We still give an additional reason. On a fixed complete null cone, the asymptotic behavior \eqref{weak} states as
\begin{align*}
\alpha,\beta,\rho,\sigma=o(r^{-5/2}),\ \betab=O(r^{-2}),\ \alphab=O(r^{-1}),
\end{align*}
The corresponding estimates for the connection coefficients\footnote{See the next section for the definitions.} state as:
\begin{equation}\label{weakconnection}
\begin{split}
\chih=o(r^{-3/2})&,\ \etab=o(r^{-3/2}),\\
\Omega\tr\chi-\overline{\Omega\tr\chi}=O(r^{-2}),\ \Omega\tr\chi&,\Omega\tr\chib=O(r^{-1}),\ \chibh=O(r^{-1}).
\end{split}
\end{equation}
Suppose that the null cone is foliated by spherical sections $S_{\ub}$ labeled by $\ub$, where $r\sim \ub$. Then we consider the Hawking mass of $S_{\ub}$:
\begin{align*}
m(\ub)=\frac{r}{2}\left(1+\frac{1}{16\pi}\int_{S_{\ub}}\tr\chi\tr\chib\D\mu_{\gs_{S_{\ub}}}\right).
\end{align*}
We can evaluate the derivative of $m(\ub)$ with respect to $\ub$:
\begin{align}\label{dmdub}
\frac{\partial m}{\partial \ub}=\frac{r}{16\pi}\int_{S_{\ub,u}}\left[(\overline{\Omega\tr\chi}-\Omega\tr\chi)\underline{\mu}-\frac{1}{2}\Omega\tr\chib|\chih|^2+\Omega\tr\chi|\etab|^2\right]\D\mu_{\gs}
\end{align}
where $\mub=-\rho+\frac{1}{2}(\chih,\chibh)-\divs\etab$ is the mass aspect function relative to a transversal null direction orthogonal to $S_{\ub}$, whose integral over $S_{\ub}$ determines the Hawking mass. By \eqref{weakconnection}, the Hawking mass $m(\ub)$ has a limit $m_\infty$ as $\ub  \text{(and $r$)} \to+\infty$. In addition, the decay conditions \eqref{weakconnection}, especially the decay conditions on $\chih$ and $\etab$, guarantee the convergence of the Hawking mass in an optimal way.

Therefore, it is natural to study the characteristic problem with the decay rates \eqref{weak}, \eqref{weakconnection}, but \textit{without} assuming smallness. The main result of this paper is the following semi-global existence result (which is a rough form of Theorem \ref{maintheorem}):
\begin{theorem}
We consider a characteristic problem, with the geometric quantities on the complete null cone behaving asymptotically like the following:
\begin{gather*}
\chih=o(r^{-3/2}),\ \zeta=o(r^{-3/2});\\
\chibh=O(r^{-1}),\ \tr\chi-\overline{\tr\chi}=O(r^{-2}),\ \tr\chib=O(r^{-1});\\
\alpha,\beta,\rho,\sigma=o(r^{-5/2});\\
\betab=O(r^{-2}),\ \alphab=O(r^{-1}).
\end{gather*}
Then the solution to the vacuum Einstein equations always exists in a uniform future neighborhood of the initial null cone. In particular, a global optical function $u$ normalized at infinity  is constructed.
\end{theorem}
\begin{remark}This theorem shows that,  the \textit{weakest decay} \eqref{weak}, \eqref{weakconnection} can also propagate locally to the future in retarded time. Recall that, if the rescaled metrics on the spherical sections $S_{\ub}$ tend to a standard round metric along the null cone, then the limit of the Hawking mass is the Bondi energy\footnote{See the discussions in \cite{Sau} and the references therein.}. In the proof of the theorem, we do \textit{not} need to assume that the spherical sections of the initial complete null cone $C_0$ tend to a standard round metric, in other word, the limit of the Hawking mass is allowed to have nothing to do with the Bondi mass. In addition, after proving the semi-global existence,  we can also prove that limits of the Hawking mass on $C_u$, which is the level sets of the global optical function $u$, decrease as a function of $u$, using a similar argument of proving the Bondi energy loss formula in \cite{K-N}.
\end{remark}

We end this section by making some comments of the proof. The basic idea follows from that of the previous work \cite{L-Z}. The major difference in the energy estimates is that the decay for $\alpha$, $\beta$, $\rho$ and $\sigma$ is weaker than that in the previous work, we use a different group of weights to generate weighted energy estimates using the null Bianchi equations. For both pairs of the equations for $\Db\alpha$-$D\beta$ and $\Db\beta$-$D(\rho,\sigma)$, we use the same weight $r^2$, which means that, we rewrite the null Bianchi equations as
\begin{gather*}
\Db(r^2|\alpha|^2\D\mu_{\gs})+D(2r^2|\beta|^2\D\mu_{\gs})=\cdots,\\
\Db(r^2|\beta|^2\D\mu_{\gs})+D(r^2(|\rho|^2+|\sigma|^2)\D\mu_{\gs})=\cdots.
\end{gather*}
The right hand side consists of a precise divergence term on the spherical sections and lower order terms. By integrating over the whole spacetime manifold, we can obtain the desired estimates for the curvature components. Recall that we use the weight $r^4$ in \cite{L-Z}. Notice that this is a refinement of the multiplier vectorfield method used in \cite{Bie}, which relies on a carefully designed group of approximately Killing and conformal Killing vectorfields. For the angular derivatives of the curvature, we use the relation $\nablas\sim r^{-1}$ as before.

As remarked before, some borderline terms appear in the energy estimates when we study the weak decay. Such terms appear in the energy estimates generated by the group $\Db(\rho,\sigma)$-$D\betab$ of the null Bianchi equations, that is, the terms with underline in the following:
\begin{align*}
\Db(r^2(|\rho|^2+|\sigma|^2)\D\mu_{\gs})+D(r^2|\betab|^2\D\mu_{\gs})=\cdots\underline{-r^2(\rho(\Omega\chih,\alphab)+\sigma(\Omega\chih\wedge\alphab))}.
\end{align*}
The $L^1$ norm of the underlined terms over the spacetime is estimated by ($\|\cdot\|$ refers to $\|\cdot\|_{L^2(S_{\ub',u'})}$)
\begin{align*}
&\int_0^{\ub}\int_0^{\varepsilon}\|r\rho\|\sum_{i=0}^2\|(r\nablas)^i\chih\|\|\alphab\|\D u'\D\ub'\\
\le&\int_0^{\ub}\sup_u\left\{\sum_{i=0}^2\|(r\nablas)^i\chih\|\right\}\left(\int_0^\varepsilon\|r\rho\|^2\D u'\right)^{1/2}\left(\int_0^{\varepsilon}\|\alphab\|^2\D u'\right)^{1/2}\D\ub'\\
\lesssim&\sum_{i=0}^2\int_0^{\ub}\sup_u\|(r\nablas)^i\chih\|^2\D\ub'\cdot\left(\int_0^\varepsilon\int_0^{\ub}\|r\rho\|^2\D\ub'\D u'\right)^{1/2}\cdot\sup_{\ub}\left(\int_0^{\varepsilon}\|\alphab\|^2\D u'\right)^{1/2}.
\end{align*}
The above three factors are expected to be bounded because the expected decay rates for $\chih$, $\rho$ and $\alphab$ are $\chih=o(r^{-3/2})$, $\rho=o(r^{-5/2})$ and $\alphab=O(r^{-1})$. We can observe that the $L^1$ norm of the underlined terms diverges if we do not use the full decay of $\chih$, $\rho$ and $\alphab$. This is exactly the reason why we call them borderline terms.

As the borderline terms appear in in the curvature estimates, similar phenomena appears in the estimate for the connection coefficients. Such terms appear in the estimates for $\nablas^3(\Omega\chih)$ and $\nablas^3(\Omega\tr\chi)$, using the coupled system of the Raychaudhuri equation and null Codazzi equation:
\begin{align}\label{Raychaudhuri}
D(\Omega\tr\chi)&=-\frac{1}{2}(\Omega\tr\chi)^2\underbrace{-|\Omega\chih|^2}_{I}\underbrace{+2\omega\Omega\tr\chi}_{II},\\
\label{nullCodazzi}
\divs(\Omega\chih)&=\frac{1}{2}\ds (\Omega\tr \chi)+\Omega\chih\cdot\etab\underbrace{-\frac{1}{2}\Omega\tr \chi\etab}_{III}-\Omega\beta.
\end{align}
First of all, we need to integrate \eqref{Raychaudhuri} from the last slice, therefore the convergence of the integral of the right hand side is essential. The coefficient $-\frac{1}{2}$ before $(\Omega\tr\chi)^2$ requires that both the terms $I$ and $II$ should behave like $o(r^{-3})$ to guarantee the convergence. 

For the term $I$, this is true because it is exactly the way what $\chih$ is expected to behave, which is one of the crucial conditions guaranteeing the convergence of the Hawking mass, see \eqref{weakconnection} and \eqref{dmdub}. Notice that the term $I$ is exactly the term $-\frac{1}{2}\Omega\tr\chib|\chih|^2$ in \eqref{dmdub}. However, the estimate for (the third order angular derivatives) of $\chih$ is derived using \eqref{nullCodazzi}. This derivation relies on $\beta$'s full decay $o(r^{-5/2})$, and $\etab$'s full decay (the term $III$) $o(r^{-3/2})$, the latter of which happens to be another crucial condition responsible for the convergence of the Hawking mass.

The behavior of the term $II$ is determined by the behavior of $\omega$, which measures the difference between the double null foliation and the geodesic foliation. In the current case, in order that $II$ behaves like $o(r^{-3})$, $\omega$ should behave at least like $o(r^{-2})$. The top order angular derivatives\footnote{If we assume up to $l$ order angular derivatives of the curvature components, then the top order angular derivatives of the connection coefficients refers to the order $l+1$.} of $\omega$ should be estimated through the equation
\begin{align*}
\Deltas\omega=\cdots+\divs(\Omega\beta),
\end{align*}
which only suggests that (the top order angular derivatives) of $\omega$ behaves like $o(r^{-3/2})$, because $\beta$ behaves only like $o(r^{-5/2})$, in the current case. Apparently, this is not enough. However, the estimate for the next to the top order angular derivatives of $\omega$ is derived through the equation
\begin{align*}
\Db\omega=\cdots+\Omega^2\rho,
\end{align*}
which suggests that $\omega$ behaves like $o(r^{-5/2})$, since $\rho$ behaves also like $o(r^{-5/2})$, and it is assumed that the initial null cone $C_0$ is foliated by geodesic foliation, that is, $\omega\equiv0$. Similar losses in the estimates for top order quantities also appears in \cite{Chr}. The above argument suggests that the top order angular derivatives of the connection coefficients may not be able to directly estimated purely under a double null foliation\footnote{However, the Raychaudhuri equation \eqref{Raychaudhuri} can be written in the following form:
\begin{align*}
\frac{\partial}{\partial s}\tr\chi'=-\frac{1}{2}(\tr\chi')^2-|\chih'|^2
\end{align*}
where $\chi'=\Omega^{-1}\chi$ and $s$ is the affine parameter. The Raychaudhuri equation has a preferable form in geodesic foliation. This suggests that the top order angular derivatives of $\tr\chi$ can be estimated, but behaves differently between the original double null foliation and geodesic foliation, see \cite{Bie} and also see \cite{Ch-K}.}. Therefore, since no smallness is assumed, it is convenient  to work under the assumption of the third order angular derivatives of the curvature components\footnote{Under such a strong regularity, many other estimates are actually easier to derive, see \cite{Luk}. It turns out that if we assume a strong regularity in the first place, many difficulties will not come out.}\footnote{By appealing to the preferable form of the Raychaudhuri equation, combining the techniques in \cite{K-R-09}, which is also commented in \cite{Luk}, we should be able to assume only the first order derivatives of the curvature. In this paper, we only focus on the problem of decay.}.

The remainder of this paper is organized as follows. Section \ref{preliminary} defines the double null foliation, the components of the connection and curvature under such foliation, and lists the null structure equations and the null Bianchi equations. Section \ref{structure} states the main theorem, which is Theorem \ref{maintheorem}, and outlines the main steps of the proof. Section \ref{estimate} is devoted to the a priori estimate, Theorem \ref{apriori}, which is the main step of the whole bootstrap argument. The Appendix sketches the proof of the construction of the canonical foliation on the last slice.

\section{Preliminary}\label{preliminary}

\subsection{Basic Geometric Setup} We follow the geometric setup and notations in \cite{Chr}.  We use $M$
 %= M(u_0, \delta + 1)$ ($u_0 \leq -2$ is a fixed negative number and $\delta$ is a small positive number which will be determined later)
 to denote the underlying space-time (which will be the solution) and use $g$ to denote the background 3+1 dimensional Lorentzian metric. We use $\nabla$ to denote the Levi-Civita connection of the metric $g$.

Let $\ub$ and $u$ be two optical functions on $M$, that is
\begin{equation*}
g(\nabla\ub,\nabla \ub)= g(\nabla u,\nabla u)=0.
\end{equation*}
The space-time $M$ is foliated by the level sets of $\ub$ and $u$ respectively. Since the gradients of $u$ and $\ub$ are null, we call the the these two foliations together a double null foliation. We require the functions $u$ and $\ub$ increase towards the future. We use $C_u$ to denote the outgoing null hypersurfaces which are the level sets of $u$ and use ${\Cb}_{\ub}$ to denote the incoming null hypersurfaces which are the level sets of $\ub$. We denote the intersection $S_{\ub,u}=\Cb_{\ub} \cap C_u$, which is a  space-like two-sphere.

We define a positive function $\Omega$ by the formula 
$$ \Omega^{-2}=-2g(\nabla\ub,\nabla u).$$
  We then define the normalized null pair $(e_3, e_4)$ by 
  $$e_3=-2\Omega\nabla\ub,\ e_4=-2\Omega\nabla u,$$ and define one another null pair 
  $$\Lb=\Omega e_3,\ L=\Omega e_4.$$
   We remark that the flows generated by $\Lb$ and $L$ preserve the double null foliation. On a given two sphere $S_{\ub, u}$ we choose a local orthonormal frame $(e_1,e_2)$. We call $(e_1, e_2, e_3, e_4)$ a \emph{null frame}.  As a convention, throughout the paper, we use capital  Latin letters $A, B, C, \cdots$ to denote an index from $1$ to $2$, e.g. $e_A$ denotes either $e_1$ or $e_2$.

We define $\phi$ to be a tangential tensorfield if $\phi$ is \textit{a priori} a tensorfield defined on the space-time $M$ and all the possible contractions of $\phi$ with either $e_3$ or $e_4$ are zeros. We use $D\phi$ and $\Db\phi$ to denote the projection to $S_{\ub,u}$ of usual Lie derivatives $\mathcal{L}_L\phi$ and $\mathcal{L}_{\Lb}\phi$. The space-time metric $g$ induces a Riemannian metric $\gs$ on $S_{\ub,u}$ and $\epsilons$ is the volume form of $\gs$ on $S_{\ub,u}$. We use $\ds$ and $\nablas$ to denote the exterior differential and covariant derivative (with respect to $\gs$) on $S_{\ub,u}$.

\begin{comment}Let $(\theta^A)_{A=1,2}$ be a local coordinate system on the two sphere $S_{0,0}$. We can extend $\theta^A$'s  to the entire $M$ by first setting $L(\theta^A)=0$ on $C_{0}$, and then setting $\Lb(\theta^A)=0$ on $M$. Therefore, we obtain a local coordinate system $(\ub,u,\theta^A)$ on $M$. In such a coordinate system, the Lorentzian metric $g$ takes the following form
\begin{align*}
g=-2\Omega^2(\D\ub\otimes\D u+\D u\otimes\D \ub)+\gs_{AB}(\D\theta^A-b^A\D\ub)\otimes(\D\theta^B-b^B\D\ub).
\end{align*}
The null vectors $\Lb$ and $L$ can be computed as $\Lb=\partial_u$ and $L=\partial_{\ub}+b^A\partial_{\theta^A}$. By construction, we have $b^A(\ub,u_0,\theta)=0$.
\end{comment}

We recall the definitions of null connection coefficients. Roughly speaking, the following quantities are Christoffel symbols of $\nabla$ according to the null frame $(e_1,e_2,e_3,e_4)$:
\begin{align*}
\chi_{AB}&=g(\nabla_Ae_4,e_B),\quad \eta_A=-\frac{1}{2}g(\nabla_3e_A,e_4),\quad \omega=\frac{1}{2}\Omega g(\nabla_4e_3,e_4),\\
\chib_{AB}&=g(\nabla_Ae_3,e_B), \quad\etab_A=-\frac{1}{2}g(\nabla_4e_A,e_3), \quad\omegab=\frac{1}{2}\Omega g(\nabla_3e_4,e_3).
\end{align*}
They are all tangential tensorfields. We also define the following normalized quantities:
$$\chi'=\Omega^{-1}\chi,\ \chib'=\Omega^{-1}\chi,\ \zeta=\frac{1}{2}(\eta-\etab).$$
 The trace of $\chi$ and $\chib$ will play an important role in Einstein field equations and they are denoted by 
 $$\tr\chi = \gs^{AB}\chi_{AB},\ \tr\chib = \gs^{AB}\chib_{AB}.$$
  By definition, we can check directly the following useful identities :
  $$\ds\log\Omega=\frac{1}{2}(\eta+\etab),\ D\log\Omega=\omega,\ \Db\log\Omega=\omegab.$$

We can also define the null components of the curvature tensor
{\bf R}:
\begin{align*}
\alpha_{AB}&=\mathbf{R}(e_A,e_4,e_B,e_4),\quad\beta_A=\frac{1}{2}\mathbf{R}(e_A,e_4,e_3,e_4),\quad\rho=\frac{1}{4}\mathbf{R}(e_3,e_4,e_3,e_4),\\
\alphab_{AB}&=\mathbf{R}(e_A,e_3,e_B,e_3),\quad\betab_A=\frac{1}{2}\mathbf{R}(e_A,e_3,e_3,e_4),\quad\sigma=\frac{1}{4}\mathbf{R}(e_3,e_4,e_A,e_B)\epsilons^{AB}.
\end{align*}

\subsection{Equations} We first define several kinds of contraction of the tangential tensorfields, which are used in expressing the equations. For a  symmetric tangential 2-tensorfield $\theta$, we use $\widehat{\theta}$ and $\tr\theta$ to denote the trace-free part and trace of $\theta$ (with respect to $\gs$). If $\theta$ is trace-free, $\Dh\theta$ and $\Dbh\theta$ refer to the trace-free part of $D\theta$ and $\Db\theta$. Let $\xi$ be a tangential $1$-form. We define some products and operators for later use. For the products, we define $(\theta_1,\theta_2)=\gs^{AC}\gs^{BD}(\theta_1)_{AB}(\theta_2)_{CD}$ and $\ (\xi_1,\xi_2)=\gs^{AB}(\xi_1)_A(\xi_2)_B$. This also leads to the following norms $|\theta|^2=(\theta,\theta)$ and $|\xi|^2=(\xi,\xi)$. We then define the contractions $(\theta\cdot\xi)_A=\theta_A{}^B\xi_B$, $(\theta_1\cdot \theta_2)_{AB}=(\theta_1)_A{}^C(\theta_2)_{CB}$, $\theta_1 \wedge\theta_2=\epsilons^{AC}\gs^{BD} (\theta_1)_{AB}(\theta_2)_{CD}$ and $\xi_1\tensor \xi_2=\xi_1\otimes\xi_2+\xi_2\otimes\xi_1-(\xi_1,\xi_2)\gs$. The Hodge dual for $\xi$ is defined by $\prescript{*}{}\xi_A=\epsilons_A{}^C\xi_C$. For the operators, we define $\divs\xi=\nablas^A\xi_A$, $\curls\xi_A=\epsilons^{AB}\nablas_A\xi_B$ and $(\divs\theta)_A=\nablas^B\theta_{AB}$. We finally define a traceless operator $(\nablas\tensor\xi)_{AB}=(\nablas\xi)_{AB}+(\nablas\xi)_{BA}-\divs\xi \,\gs_{AB}$.

%For the sake of simplicity, we will use abbreviations $\Gamma$ and $R$ to denote an arbitrary connection coefficient and an arbitrary null curvature component. We introduce a schematic way to write products. Let $\phi$ and $\psi$ be arbitrary tangential tensorfields, we also use $\phi\cdot\psi$ to denote an arbitrary contraction of $\phi$ and $\psi$ by $\gs$ and $\epsilons$. This schematic notation only captures the quadratic nature of the product, and it will be good enough for most of the cases when we derive estimates.

The following is the null structure equations that are used in this paper. (where  $K$ is the Gauss curvature of $S_{\ub,u}$):\footnote{See Chapter 1 of \cite{Chr} for the derivation of these equations.}
\begin{align*}
D(\Omega\tr\chi)&=-\frac{1}{2}(\Omega\tr\chi)^2-|\Omega\chih|^2+2\omega\Omega\tr\chi,\\
\Dbh\chibh'&=-\alphab,\\
\Db\tr\chib'&=-\frac{1}{2}\Omega^2(\tr\chib')^2-\Omega^2|\chibh'|^2,\\
D\eta &= \Omega(\chi \cdot\etab-\beta),\\
\Db\etab &= \Omega(\chib \cdot\eta+\betab),\\
D  \omegab &=\Omega^2(2(\eta,\etab)-|\eta|^2-\rho),\\
\Db  \omega &=\Omega^2(2(\eta,\etab)-|\etab|^2-\rho),\\
K&=-\frac{1}{4}\tr \chi\tr\chib+\frac{1}{2}(\chih,\chibh)-\rho,\\
\divs(\Omega\chih)&=\frac{1}{2}\ds (\Omega\tr \chi)+\Omega\chih\cdot\etab-\frac{1}{2}\Omega\tr \chi\etab-\Omega\beta,\\
\Dbh(\Omega\chih)&=\Omega^2(\nablas \tensor \eta + \eta \tensor \eta +\frac{1}{2}\tr\chib\chih-\frac{1}{2}\tr\chi \chibh),\\
\Db(\Omega\tr\chi)&=\Omega^2(2\divs\eta+2|\eta|^2-(\chih,\chibh)-\frac{1}{2}\tr\chi\tr\chib+2\rho),\\
\Db\eta&=-\Omega(\chib\cdot\eta+\betab)+2\ds\omegab.
\end{align*}
We also use the null frame to decompose the contracted second Bianchi identity $\nabla^\alpha \mathbf{R}_{\alpha\beta\gamma\delta} = 0$ into components. This leads the following null Bianchi equations:\footnote{See Proposition 1.2 of \cite{Chr}.}
\begin{gather*}
\Dbh\alpha-\frac{1}{2}\Omega\tr\chib \alpha+2\omegab\alpha+\Omega\{-\nablas\tensor\beta -(4\eta+\zeta)\tensor \beta+3\chih \rho+3{}^*\chih \sigma\}=0,\\
\Dh\alphab-\frac{1}{2}\Omega\tr\chi \alphab+2\omega\alphab+\Omega\{\nablas\tensor\betab +(4\etab-\zeta)\tensor \betab+3\chibh \rho-3{}^*\chibh \sigma\}=0,\\
D\beta+\frac{3}{2}\Omega\tr\chi\beta-\Omega\chih\cdot\beta-\omega\beta-\Omega\{\divs\alpha+(\etab+2\zeta)\cdot\alpha\}=0,\\
\Db\betab+\frac{3}{2}\Omega\tr\chib\betab-\Omega\chibh\cdot\betab-\omegab\betab+\Omega\{\divs\alphab+(\eta-2\zeta)\cdot\alphab\}=0,\\
\Db\beta+\frac{1}{2}\Omega\tr\chib\beta-\Omega\chibh \cdot \beta+\omegab \beta-\Omega\{\ds \rho+{}^*\ds \sigma+3\eta\rho+3{}^*\eta\sigma+2\chih\cdot\betab\}=0,\\
D\betab+\frac{1}{2}\Omega\tr\chi\betab-\Omega\chih \cdot \betab+\omega \betab+\Omega\{\ds \rho-{}^*\ds \sigma+3\etab\rho-3{}^*\etab\sigma-2\chibh\cdot\beta\}=0,\\
D\rho+\frac{3}{2}\Omega\tr\chi \rho-\Omega\{\divs \beta+(2\etab+\zeta,\beta)-\frac{1}{2}(\chibh,\alpha)\}=0,\\
\Db\rho+\frac{3}{2}\Omega\tr\chib \rho+\Omega\{\divs \betab+(2\eta-\zeta,\betab)+\frac{1}{2}(\chih,\alphab)\}=0,\\
D\sigma+\frac{3}{2}\Omega\tr\chi\sigma+\Omega\{\curls\beta+(2\etab+\zeta,{}^*\beta)-\frac{1}{2}\chibh\wedge\alpha\}=0,\\
\Db\sigma+\frac{3}{2}\Omega\tr\chib\sigma+\Omega\{\curls\betab+(2\etab-\zeta,{}^*\betab)+\frac{1}{2}\chih\wedge\alphab\}=0.
\end{gather*}

We denote the first order elliptic operators (or Hodge operators) appearing above in the null Bianchi equations, by $\mathcal{D}_1,\mathcal{D}_2$, where
\begin{align*}
&\mathcal{D}_1:\text{tangential one-form $\xi\mapsto$ a pair of functions $(\divs\xi,\curls\xi)$};\\
&\mathcal{D}_2:\text{tangential symmetric trace-free $(0,2)$ type tensorfield $\theta\mapsto$ tangential one-form $\divs\theta$}.
\end{align*}
It is easy to calculate the formal $L^2$ adjoint
\begin{align*}
&^*\mathcal{D}_1:\text{a pair of functions $(f,g)\mapsto$ tangential one-form $-\ds f+{}^*\ds g$};\\
&^*\mathcal{D}_2:\text{tangential one-form $\xi\mapsto$ tangential symmetric trace-free $(0,2)$ type tensorfield $-\frac{1}{2}\nablas\tensor\xi$}.
\end{align*}
We will denote any one of the above elliptic operators (or Hodge operators) and their formal $L^2$ adjoint by $\mathcal{D},{}^*\mathcal{D}$.

Before proceeding further, we list the commutation formulas which are used for the estimates of derivatives\footnote{See Chapter 4 of \cite{Chr} for the first group. The second group can be derived directly by the definition of curvature.}.
\begin{lemma}\label{commutator}
Given integer $i$ and tangential tensorfield $\phi$. we have
\begin{align*}
[D,\nablas^i]\phi&=\sum_{j=1}^i\nablas^j(\Omega\chi)\cdot\nablas^{i-j}\phi,\\
[\Db,\nablas^i]\phi&=\sum_{j=1}^i\nablas^j(\Omega\chib)\cdot\nablas^{i-j}\phi,
\end{align*}
and
\begin{align*}
[\mathcal{D},\nablas^i]\phi&=\sum_{j=1}^i\nablas^{j-1}K\cdot\nablas^{i-j}\phi,\\
[^*\mathcal{D},\nablas^i]\phi&=\sum_{j=1}^i\nablas^{j-1}K\cdot\nablas^{i-j}\phi.
\end{align*}
Here we use ``$\cdot$'' to represent an arbitrary contraction with the coefficients by $\gs$ or $\epsilons$. In addition, if $\phi$ is a function, then when $i=1$, all commutators above are zero; when $i\ge2$, all $i$'s are replaced by $i-1$'s in above formulas.
\end{lemma}

\section{Main Theorem and Structure of the Proof}\label{structure}

\subsection{The Statement of Main Theorem}\label{maintheoremsection}

We first introduce the following scale invariant norms (for $2\le p\le\infty$ and $q=1,2$):
\begin{gather*}
\|\xi\|_{\mathbb{L}^p(\ub,u)}=\left(\int_{S_{\ub,u}}r^{-2}|\xi|^p\D\mu_{\gs}\right)^{\frac{1}{p}},\\
\|\xi\|_{\mathbb{H}^n(\ub,u)}=\sum_{i=0}^n\|(r\nablas)^i\xi\|_{\L^p(\ub,u)},\\%\left(\int_{S_{\ub,u}}r^{-2}|r^i\nablas^i\xi|^2\D\mu_{\gs}\right)^{\frac{1}{2}},\\
\|\xi\|_{\mathbb{L}^q_{[\ub_1,\ub_2]}\L^p(u)}=\left(\int_{\ub_1}^{\ub_2}r^{-1}\|\xi\|_{\L^p(\ub',u)}^q%\left(\int_{S_{\ub',u}}r^{-2}|\xi|^p\D\mu_{\gs}\right)^{\frac{q}{p}}
\D\ub'\right)^{\frac{1}{q}},\\
\|\xi\|_{\mathbb{L}^q_{[\ub_1,\ub_2]}\mathbb{H}^n(u)}=%\sum_{i=0}^n 
 \left(\int_{\ub_1}^{\ub_2}r^{-1}\|\xi\|_{\H^n(\ub',u)}^q%\left(\int_{S_{\ub',u}}r^{-2}|r^i\nablas^i\xi|^2\D\mu_{\gs}\right)^{\frac{q}{2}} 
 \D\ub'\right)^{\frac{1}{q}},\\
\|\xi\|_{\L^q_u\L^p(\ub)}=\left(\int_{0}^{\varepsilon}\|\xi\|_{\L^p(\ub,u')}^q%\left(\int_{S_{\ub,u'}}r^{-2}|\xi|^p\D\mu_{\gs}\right)^{\frac{q}{p}}
\D u'\right)^{\frac{1}{q}},\\
\|\xi\|_{\L^q_u\mathbb{H}^n(\ub)}=%\sum_{i=0}^n
\left(\int_{0}^{\varepsilon}\|\xi\|_{\H^n(\ub,u')}^q%\left(\int_{S_{\ub,u'}}r^{-2}|r^i\nablas^i\xi|^2\D\mu_{\gs}\right)^{\frac{q}{2}}
\D u'\right)^{\frac{1}{q}}.
\end{gather*}
In addition, we define
\begin{align*}
\|\xi\|_{\L^q_{[\ub_1,\ub_2]}\L^\infty_{u}\mathbb{H}^n}=%\sum_{i=0}^n
\left(\int_{\ub_1}^{\ub_2}\sup_{u'\in[0,\varepsilon]}r^{-1}\|\xi\|_{\H^n(\ub',u')}^q%\left(\int_{S_{\ub',u}}r^{-2}|r^i\nablas^i\xi|^2\D\mu_{\gs}\right)^{\frac{q}{2}}
\D \ub'\right)^{\frac{1}{q}}
\end{align*}

The main theorem of this paper is the following.
\begin{theorem}[Main Theorem]\label{maintheorem}
Suppose that $C_0$ and $\Cb_0$ are two intersecting null cones where $C_0$ is outgoing complete and $\Cb_0$ is incoming, and $S_0=\Cb_0\bigcap C_0$ is a two-sphere. Suppose also that $C_0$ and $\Cb_0$ are foliated by affine sections which are labelled by two functions $s$ and $\underline{s}$. Let $\Lambda(s)$ and $\lambda(s)$ be the larger and smaller eigenvalue of $r(s)^{-2}\gs|_{S_{s,0}}$ with respect to $r(0)^{-2}\gs|_{S_{0,0}}$, and $r(s)$ be the area radius defined by $4\pi r(s)^2=Area(S_{s,0})$. If the initial data given on $\Cb_0\bigcup C_0$ satisfy the following:
\begin{align*}
\mathcal{O}_0=&\sup_{s}\left\{\|r^2\widetilde{\tr\chi},r^{3/2}\chih,r\tr\chi\|_{\H^2(s,0)}+\|r^{3/2}\zeta\|_{\H^3(s,0)}+r\|\chibh,\tr\chib\|_{\H^3(s,0)}\right\}\\
&+\sup_s\left\{\|r^{3/2}\chih,r^{3/2}\zeta\|_{\L^2_{\ub}\H^2(s)}\right\}\\
&+\sup_s\Lambda(s)+(\inf_s\lambda(s))^{-1}+\sup_s\frac{r}{1+s}+\left(\inf_s\frac{r}{1+s}\right)^{-1}<\infty\\
\mathcal{R}_0=&\|r^{5/2}\alpha,r^{5/2}\beta,r^{5/2}\rho,r^{5/2}\sigma,r^{3/2}\betab\|_{\L^2_{[0,+\infty)}\H^3(0)}\\
&+\|r^2\beta,r^2\rho,r^2\sigma,r^2\betab,r\alphab\|_{\L^2_{u}\H^3(0)}\\
&+\sup_{s}\|r^{5/2}\beta,r^{5/2}\rho,r^{5/2}\sigma,r^2\betab\|_{\H^2(s,0)}\\
&+\sup_{\underline{s}}\left\{\|r\alphab\|_{\H^2(0,\underline{s})}+\|r\Db\alphab\|_{\H^1(0,\underline{s})}+\|r\Db^2\alphab\|_{\L^2(0,\underline{s})}\right\}<\infty.
\end{align*}
Then there exists an $\varepsilon>0$ depends on $\mathcal{O}_0$, $\mathcal{R}_0$ and a global optical function $u$ such that the solution of vacuum Einstein equations exists in a global double null foliation
$0\le\ub<+\infty$, $0\le u\le\varepsilon$. In addition, $\Omega\to1$ as $\ub\to+\infty$.
\end{theorem}
\begin{remark}
The last statement that $\Omega\to1$ as $\ub\to+\infty$ means that $u$ tends to the affine parameter of the incoming null geodesic as $\ub\to+\infty$. This guarantees that the null cones $C_u$ will not shrink when approaching to infinity.
\end{remark}

\subsection{Structure of the Proof}

We will prove the above Main Theorem in the remaining part of this paper.  The structure of the proof is similar to that in \cite{L-Z}. We begin the proof by defining
$$\mathcal{A}_{\varepsilon,\Delta}=\{c\ge0: \text{$c$ satisfies the following two properties for small $\varepsilon>0$ and large $\Delta>0$}\},$$
where $\varepsilon$ is a small positive parameter and $\Delta$ is a positive large constant. They will be suitably chosen in the context depending only on $\mathcal{O}_0,\mathcal{R}_0$.
\begin{itemize}
\item[(1)]The solution of vacuum Einstein equations $g$ exists in a double null foliation given by $(\ub,u)$ for $0\le\ub\le c$, $0\le u\le\varepsilon$, where $\ub=s$ on $C_0$ and $u$ on $\Cb_{\ub_*}$ defines a canonical foliation, which means the following equation is satisfied:\begin{align}\label{lastslice}
\overline{\log\Omega}=0,\quad \Deltas\log\Omega=\frac{1}{2}\divs\etab+\frac{1}{2}\left(\frac{1}{2}((\chih,\chibh)-\overline{(\chih,\chibh)})-(\rho-\overline{\rho})\right),
\end{align}
and $u=0$ on $S_{0,0}$.
\item[(2)]Written in the double null foliation given by $(\ub,u)$, $\mathcal{R}\le\Delta$, where 
\begin{align*}
\mathcal{R}=&\sup_{u}\left\{\|r^{5/2}\alpha,r^{5/2}\beta,r^{5/2}\rho,r^{5/2}\sigma,r^{3/2}\betab\|_{\L^2_{[0,\ub_*]}\H^3(u)}\right\}\\
&+\sup_{\ub}\left\{\|r^2\beta,r^2\rho,r^2\sigma,r^2\betab,r\alphab\|_{\L^2_{u}\H^3(\ub)}\right\}
\end{align*}
\end{itemize}
We will prove, for $\varepsilon>0$ sufficiently small and $\Delta$ sufficiently large, $\ub_*=\sup\mathcal{A}_{\varepsilon,\Delta}=+\infty$.

The proof is divided into following steps:
\begin{itemize}
\item[\bf Step 1.] In this step, the canonical foliation on $\Cb_0$ is constructed and therefore $\mathcal{A}_{\varepsilon,\Delta}$ is not empty. Then we assume that $\ub_*=\sup\mathcal{A}<+\infty$. %The proof is basically as in Section \ref{SectionLastSlice}. In this case, the condition $W(s,\theta)\le\varepsilon+\delta'$ is not needed in the definition of the function space $\mathcal{K}$, and the proof is slightly easier. We then proceed as in Step 4, showing that $\mathcal{A}$ is not empty. We then assume $\ub_*=\sup\mathcal{A}<+\infty$.
\item[\bf Step 2.] This step is devoted to the a priori estimate. It is not hard to see $\ub_*\in\mathcal{A}_{\varepsilon,\Delta}$. We work on the space-time region $M_{\ub_*,\varepsilon}$ which corresponds to $0\le\ub\le\ub_*$, $0\le u\le\varepsilon$. and then the function $u$ restricted on $\Cb_{\ub_*}$ induces a canonical foliation. We will prove that, if $\varepsilon>0$ is sufficiently small,
$$\mathcal{R}\le C(\mathcal{O}_0,\mathcal{R}_0).$$
In particular, we can choose $\Delta$ sufficiently large such that $\mathcal{R}\le\frac{1}{4}\Delta$.
\item[\bf Step 3.]By the existence result in \cite{Luk}, We extend the solution $g$ to $0\le\ub\le\ub_*+\delta$, $0\le u\le \varepsilon+\delta'$ for $\delta,\delta'$ sufficiently small, and then construct a new optical function $u_\delta$ varying in $[0,\varepsilon]$, such that $(\ub,u_\delta)$ is a new double null foliation, and $u_\delta$ induces a canonical foliation on $\Cb_{\ub_*+\delta}$. In addition, the norms $\mathcal{R}$ expressed in the new foliation are bounded by $\frac{1}{2}\Delta$. 
\end{itemize}
Therefore, we have $\ub_*+\delta\in \mathcal{A}_{\varepsilon,\Delta}$, which contradicts to that $\ub_*=\sup\mathcal{A}_{\varepsilon,\Delta}<+\infty$. Finally, a global optical function $u$ can be constructed by a limiting argument.

We will complete Step 2 in the next section. The proof of Step 1 and Step 3, which is the construction of the canonical foliation, will be sketched in the Appendix and the reader can refers to \cite{L-Z} and the references therein for the full details.

\section{The A Priori Estimate}\label{estimate}
We define the following norms of the connection coefficients:
\begin{align*}
\mathcal{O}=&\sup_{\ub,u}\left\{r^2\|\widetilde{\Omega\tr\chi},\omega\|_{\H^3(\ub,u)}+r^{3/2}\|\chih,\eta,\etab\|_{\H^3(\ub,u)}+r\|\chibh,\tr\chi,\tr\chib,\omegab\|_{\H^3(\ub,u)}\right\}\\
&+\|r^{3/2}\chih\|_{\L^2_{[0,\ub_*]}\L^\infty_u\H^3(u)}+\|r^{3/2}\etab\|_{\L^2_{[0,\ub_*]}\L^\infty_u\H^2(u)}+\sup_u\|r^{5/2}\omega\|_{\L^2_{[0,\ub_*]}\H^3(u)}
\end{align*}

We will prove the following theorem, which is the content of Step 2.
\begin{theorem}\label{apriori}
Suppose that the assumptions in Theorem \ref{maintheorem} hold and the solution of the vacuum Einstein equations $g$ exists in $0\le\ub\le\ub_*$, $0\le u\le \varepsilon$ for some double null foliation given by two optical function $\ub,u$, and $u$ induces a cononical foliation on $\Cb_{\ub_*}$. Then if $\varepsilon$ is sufficiently small depending on $\mathcal{O}_0$ and $\mathcal{R}_0$, we have the estimates
\begin{align*}
\mathcal{O},\mathcal{R}\le C(\mathcal{O}_0,\mathcal{R}_0).
\end{align*}
\end{theorem}

The proof of the above theorem is divided into the following two propositions.
\begin{proposition}\label{connection}
Suppose that the assumptions in Theorem \ref{apriori} hold and $\mathcal{R}<\infty$. Then for $\varepsilon$ sufficiently small,
$$\mathcal{O}\le C(\mathcal{O}_0,\mathcal{R}_0,\mathcal{R}).$$
In particular, the estimates for $r^2\|\widetilde{\Omega\tr\chi},\omega\|_{\H^3(\ub,u)}$ and $r^{3/2}\|\Omega\chih\|_{\H^3(\ub,u)}$ does not depend on $\mathcal{R}$.
\end{proposition}

\begin{proposition}\label{curvaturecompletenullcone}
If $\varepsilon$ is sufficiently small, we have
\begin{align*}
\mathcal{R}\le C(\mathcal{O}_0,\mathcal{R}_0).
\end{align*}
\end{proposition}

\subsection{Proof of Proposition \ref{connection}}\label{SectionConnection}

It suffices  to prove the Proposition \ref{connection} under the following:
\begin{align*}\textbf{ Bootstrap Assumption:}\ \ \ \mathcal{O}\le \Delta_0\end{align*}
where $\Delta_0$ is a sufficiently large number.

Recall $r$ is the area radius and let $I(S_{\ub,u})$ be the isoperimetric constant of the spherical sections $S_{\ub,u}$, we introduce an auxiliary bootstrap assumption:
\begin{align}\label{geometricbootstrap}\frac{1}{4}\mathcal{O}_0^{-1}(1+\ub)\le r\le 4\mathcal{O}_0(1+\ub),\quad \frac{1}{4}\mathcal{O}_0^{-1}\le 2\pi I(S_{\ub,u})\le 4\mathcal{O}_0.\end{align}

\subsubsection{Premilary Lemmas}
Using the bound of the isoperimetric constant, we have the following Sobolev inequalities.
\begin{lemma}[Sobolev inequalities, see Section 5.2 of \cite{Chr}]\label{Sobolev}
Given a tangential tensorfield $\phi$, we have for $q\in(2,+\infty)$,
\begin{align*}
\|\phi\|_{\L^q(\ub,u)}\lesssim\|\phi\|_{\H^1(\ub,u)},\\
\|\phi\|_{\L^\infty(\ub,u)}\lesssim\|\phi\|_{\H^2(\ub,u)},
\end{align*}
\end{lemma}
We remark that the notation $A\lesssim B$ means $A\le C(\mathcal{O}_0,\mathcal{R}_0) B$. By the Sobolev inequalities, we can bound $|r\Omega\tr\chib|,|r\Omega\chibh|\le C(\mathcal{O}_0)$. Then it is a direct consequence that, if $\varepsilon$ is sufficiently small, \eqref{geometricbootstrap} holds with $2\mathcal{O}_0$ in the role of $\mathcal{O}_0$. Therefore \eqref{geometricbootstrap} actually holds.

Using the bounds of isoperimetric constant and Gauss curvature $K$, we have the elliptic estimates. 
\begin{lemma}[Elliptic estimates for Hodge systems, see Chapter 7 of \cite{Chr}]\label{elliptic}Suppose that
\begin{align*}\|r^2K\|_{\H^1(\ub,u)}\le C.\end{align*}

Now assume that $\theta$ is a tangential symmetric trace-free $(0,2)$ type tensorfield with
$$\divs\theta=f,$$
where $f$ is a tangential one-form. Then
\begin{align*}
\|\theta\|_{\H^3(\ub,u)}\lesssim\|rf\|_{\H^2(\ub,u)}+\|\theta\|_{\L^2(\ub,u)}.
\end{align*}

Assume that $\phi$ is a function with
\begin{equation*}
\Deltas\phi=f.
\end{equation*}
We have
\begin{align*}
\|\phi\|_{\H^3(\ub,u)} \lesssim\|r^2f\|_{\H^1(\ub,u)}+\|\overline{\phi}\|_{\L^2(\ub,u)}.
\end{align*}
Here $\overline{\phi}$ is the average of $\phi$ over $S_{\ub,u}$.
\end{lemma}
\begin{remark}
We remark we do not need to estimate $\|\phi\|_{\L^2(\ub,u)}$ first to apply the second part of the above lemma, which is not so in applying the first part. We avoid using the positivity of the Gauss curvature $K$.
\end{remark}

We also have the following Gronwall type estimates:
\begin{lemma}[Gronwall type estimates, see Chapter 4 of \cite{Chr} or Chapter 4 of \cite{K-N}]\label{evolution}
Suppose that there exists a constant $c$ such that $r^{3/2}|\Omega\chih|,r^{3/2}|\widetilde{\Omega\tr\chi}|\le c$. Then for an $s$-covariant tengential tensorfield $\phi$, $2\le q\le+\infty$, and any real $\nu$, we have
\begin{align*}
\|r^{s-\nu}\phi\|_{\L^q(\ub,u)}\lesssim C_{q,\nu,s}\left(\|r^{s-\nu}\phi\|_{\L^q(\ub_*,u)}+\int_{\ub}^{\ub_*}\|r^{s-\nu}(D\phi-\frac{\nu}{2}\Omega\tr\chi\phi)\|_{\L^q(\ub',u)}\D\ub'\right),
\end{align*}
and for $\varepsilon$ sufficiently small, we have
\begin{align*}
\|\phi\|_{\L^q(\ub,u)}
\lesssim C_{q}\left(\|\phi\|_{\L^q(\ub,0)}+\|\Db\phi\|_{\L_u^1\L^q(\ub)}\right).
\end{align*}
\end{lemma}

In addition, we have the Gronwall type estimate for the derivatives along $\Db$ direction:
\begin{lemma}\label{Gronwallderivative}
For $\varepsilon$ sufficiently small, we have, for $n=2,3$,
\begin{align*}
\|\phi\|_{\H^n(\ub,u)}
\lesssim \|\phi\|_{\H^n(\ub,0)}+\|\Db\phi\|_{\L_u^1\H^n(\ub)}.
\end{align*}
\end{lemma}
\begin{proof}
We apply the above Gronwall estimate on the following equation (for $1\le i\le3$)
\begin{align*}
\Db\nablas^i\phi=\nablas^i\Db\phi+[\Db,\nablas^i]\phi=\nablas^i\Db\phi+\sum_{j=1}^i\nablas^j(\Omega\chib)\nablas^{i-j}\phi.
\end{align*}
By H\"older inequality and Sobolev inequalies, we have
\begin{align*}
\|\phi_1\cdot\phi_2\|_{\H^i(\ub,u)}\lesssim\|\phi_1\|_{\H^i(\ub,u)}\|\phi_2\|_{\H^i(\ub,u)}
\end{align*}
for $i\ge2$. Therefore,
\begin{align*}
\|\phi\|_{\H^n(\ub,u)}
\lesssim \|\phi\|_{\H^n(\ub,0)}+\|\Db\phi\|_{\L_u^1\H^n(\ub)}+\varepsilon\sup_u\|\Omega\chib\|_{\H^n(\ub,u)}\sup_u\|\phi\|_{\H^n(\ub,u)}.
\end{align*}
The conclusion follows by choosing $\varepsilon$ sufficiently small.
\end{proof}
We remark that the above lemma holds with $r^\mu\phi$ in the role of $\phi$, $r^\mu\Db\phi$ in the role of $\Db\phi$, for any $\mu$ with the right hand side being finite, and then the constant depends on $\mu$.

We also have the following:
\begin{lemma}\label{L2LinftyH3}
For $\varepsilon$ sufficiently small, we have, for $n=2$,
\begin{align*}
\|r^\mu\phi\|_{\L^2_{[\ub_1,\ub_2]}\L^\infty_{[0,u]}\H^n}\lesssim_\mu \|r^\mu\phi\|_{\L^2_{[\ub_1,\ub_2]}\H^n(0)}+\varepsilon^{1/2}\left(\int_0^u\int_{\ub_1}^{\ub_2}\|r^\mu\Db\phi\|_{\H^2(\ub,u')}^2\D u'\D\ub\right)^{1/2}.
\end{align*}
for any $\mu$ with the right hand side being finite.
\end{lemma}
\begin{proof}
By Lemma \ref{Gronwallderivative}, we have
\begin{align*}
\|r^{\mu-1/2}\phi\|_{\H^n(\ub,u)} \lesssim_{\mu} \|r^{\mu-1/2}\phi\|_{\H^n(\ub,0)}+\|r^{\mu-1/2}\Db\phi\|_{\L_u^1\H^n(\ub)}.
\end{align*}
Notice that the right hand side does not depend on $u$, we then take supremum with respect to $u$, and integrate the square of both sides over $[\ub_1,\ub_2]$, we obtain
\begin{align*}
\int_{\ub_1}^{\ub_2}\sup_u\|r^{\mu-1/2}\phi\|^2_{\H^n(\ub,u)}\D\ub\lesssim_{\mu}\int_{\ub_1}^{\ub_2}\|r^{\mu-1/2}\phi\|^2_{\H^n(\ub,0)}\D\ub+\int_{\ub_1}^{\ub_2}\left(\int_0^\varepsilon\|r^{\mu-1/2}\Db\phi\|_{\H^n(\ub,u)}\D u\right)^2\D\ub.
\end{align*}
The second term in the right hand side is estimated by
\begin{align*}
\int_{\ub_1}^{\ub_2}\left(\int_0^\varepsilon\|r^{\mu-1/2}\Db\phi\|_{\H^n(\ub,u)}\D u\right)^2\D\ub\le&\varepsilon\int_{\ub_1}^{\ub_2}\int_0^\varepsilon\|r^{\mu-1/2}\Db\phi\|^2_{\H^n(\ub,u)}\D u\D\ub
\\=&\varepsilon\int_0^\varepsilon\int_{\ub_1}^{\ub_2}\|r^{\mu-1/2}\Db\phi\|^2_{\H^n(\ub,u)}\D \ub\D u.
\end{align*}
The conclusion then follows.
\end{proof}
It is crucial here we can change the order of the double integration of the second term on the right hand side. But, in this paper, we only need to estimate
\begin{align*}
\varepsilon^{1/2}\left(\int_0^u\int_{\ub_1}^{\ub_2}r^{-1}\|r^\mu\Db\phi\|_{\H^2(\ub,u')}^2\D u'\D\ub\right)^{1/2}\le\varepsilon\sup_u\|r^\mu\Db\phi\|_{\L^2_{[\ub_1,\ub_2]}\H^2(u')}.
\end{align*} 
Notice that $\|r^\mu\phi\|_{\L^2_{[\ub_1,\ub_2]}\H^n(u)}\le\|r^\mu\phi\|_{\L^2_{[\ub_1,\ub_2]}\L^\infty_u\H^n}$ and therefore we have
\begin{align*}
\|r^\mu\phi\|_{\L^2_{[\ub_1,\ub_2]}\H^n(u)}\lesssim_\mu \|r^\mu\phi\|_{\L^2_{[\ub_1,\ub_2]}\H^n(0)}+\varepsilon\sup_u\|r^\mu\Db\phi\|_{\L^2_{[\ub_1,\ub_2]}\H^2(u)}^2.
\end{align*}
We will also call the above two lemmas the Gronwall type estimates.

We first establish the following lemma, which says that the geometric quantities share the same estimates up to a multiple by $\Omega$:
\begin{lemma}If $\varepsilon>0$ is sufficiently small, then
\begin{align*}
 &\|(r\nablas)^{\le1}\log\Omega\|_{\L^\infty(\ub,u)}+ \|(r\nablas)^2\log\Omega\|_{\L^4(\ub,u)}+\|(r\nablas)^{3}\log\Omega\|_{\L^2(\ub,u)}\le r^{-1/2}C(\mathcal{O}_0,\mathcal{R}_0).
\end{align*}
In particular, we have $C(\mathcal{O}_0,\mathcal{R}_0)^{-1}\le r^{1/2}(\Omega-1)\le C(\mathcal{O}_0,\mathcal{R}_0)$. \end{lemma}
\begin{proof}Recall that on $S_{\ub_*,0}$, $\Omega$ satisfies the equation
\begin{align*}
\overline{\log\Omega}=0,\quad \Deltas\log\Omega=\frac{1}{2}\divs\etab+\frac{1}{2}\left(\frac{1}{2}((\chih,\chibh)-\overline{(\chih,\chibh)})-(\rho-\overline{\rho})\right).
\end{align*}
The expression on the right hand side is invariant, in particular, does not depend on $\Omega$. Therefore, by $L^2$ elliptic estimate and Sobolev inequalities, we have
\begin{align*}
 &\|(r\nablas)^{\le1}\log\Omega\|_{\L^\infty(\ub_*,0)}+ \|(r\nablas)^2\log\Omega\|_{\L^4(\ub_*,0)}+\|(r\nablas)^{3}\log\Omega\|_{\L^2(\ub_*,0)}\le r^{-1/2}|_{S_{\ub_*,0}}C(\mathcal{O}_0,\mathcal{R}_0).
\end{align*}
Because we do not change the foliation $\ub$ on $C_0$, $\Omega$ is extended as a constant along every null generator of $C_0$. The above estimates then hold along the whole $C_0$, i.e., replacing $\ub_*$ by $\ub$ for all $0\le\ub\le\ub_*$.

Then we can estimate $\|\log\Omega\|_{L^\infty(S_{\ub,u})}$ by the equation $\Db\log\Omega=\omegab$ as follows:
\begin{align*}
|\log\Omega|\le|\log\Omega|_{S_{\ub,0}}|+\int_0^u|\omegab|\D u'\lesssim r^{-1/2}|_{S_{\ub_*,0}}C(\mathcal{O}_0,\mathcal{R}_0)+\varepsilon r^{-1}\Delta_0
\end{align*}
and then multiply both sides by $r^{1/2}$ and choose $\varepsilon$ sufficiently small. The derivatives of $\log\Omega$ can be estimated by commuting $\nablas$ three times to the equation $\Db\log\Omega=\omegab$ and then apply Lemma \ref{Gronwallderivative} to obtain the conclusion.\end{proof}

\begin{remark}We make an important remark here. Because we  construct the canonical foliation on the last slice $\Cb_{\ub_*}$, $\Omega$ will not be constant $1$ but only constant along the null generators of $C_0$. Therefore, the norms included in $\mathcal{O}$, $\mathcal{R}$ written on $u=0$ are different from the corresponding norms in $\mathcal{O}_0$, $\mathcal{R}_0$. Fortunately, the above estimates for up to the third order derivatives of $\Omega$ ensure that the differences are up to a constant depending only on $\mathcal{O}_0$, $\mathcal{R}_0$, so we do not need to worry about this. In addition, the differences of some components do no exist, for example, $\Omega\chi$ (but not $\chi$ itself), $\chib'$, $\omega$, $\etab$, $\rho$ and $\sigma$. 
\end{remark}

\subsubsection{Estimates for the connection coefficients}
Now we turn to the estimates for the connection coefficients. For the structure equations along $\Db$ direction, we will apply Lemma \ref{Gronwallderivative} or Lemma \ref{L2LinftyH3} for $\phi$ equals to up to the second or third order angular derivatives of the connection coefficients. 

As the first step, we consider the structure equations for $\Db\chibh'$, $\Db\tr\chib'$ and $\Db\etab$. By commuting $\nablas$ at most three times to those equations, we have
\begin{align*}
\|r\chibh',r\tr\chib',r^{3/2}\etab\|_{\H^3(\ub,u)}\lesssim&C(\mathcal{O}_0)+\|r\alphab,r^{3/2}\betab\|_{\L^1_u\H^3(\ub)}+\|r\chibh'\cdot\chibh',r(\tr\chib)^2,r^{3/2}\Omega\chib\cdot\eta\|_{\L^1_u\H^3(\ub)}\\\lesssim&C(\mathcal{O}_0)+\varepsilon^{1/2}\|r\alphab,r^2\betab\|_{\L^2_u\H^3(\ub)}+\varepsilon\Delta_0^2\\
\lesssim&C(\mathcal{O}_0)
\end{align*}
if $\varepsilon$ is sufficiently small.

By applying Lemma \ref{L2LinftyH3}, we also have
\begin{equation}\label{L2Linftyeta}
\begin{split}
\|r^{3/2}\etab\|_{\L^2_{[0,\ub_*]}\L^\infty_u\H^2(u)}\lesssim&C(\mathcal{O}_0)+\varepsilon\sup_{\ub,u}\|r^{1/2}\Omega\chib\|_{\H^2(\ub,u)}\sup_u\|r\eta\|_{\L^2_{[0,\ub_*]}\H^2(u)}\\
&+\varepsilon\sup_{u}\|r^{3/2}\betab\|_{\L^2_{[0,\ub_*]}\H^2(u)}\\
\lesssim&C(\mathcal{O}_0)+\varepsilon\sup_{\ub,u}\left\{\|r^{1/2}\Omega\chib\|_{\H^2(\ub,u)}\|r^{3/2}\eta\|_{\H^2(\ub,u)}\right\}\\
&+\varepsilon\sup_{u}\|r^{3/2}\betab\|_{\L^2_{[0,\ub_*]}\H^2(u)}\\
\lesssim& C(\mathcal{O}_0)
\end{split}
\end{equation}
if $\varepsilon$ is sufficiently small. Here we use
\begin{align*}
\|r^s\phi\|_{\L^2_{[0,\ub_*]}\H^2(u)}\le\sup_{\ub}\|r^{s+1/2}\phi\|_{\H^2(\ub,u)}
\end{align*}
which is derived by the H\"older inequality.

We then consider the structure equation
\begin{align*}\Db  \omega &=\Omega^2(2(\eta,\etab)-|\etab|^2-\rho)\end{align*}
We also have
\begin{align*}
\|r^2\omega\|_{\H^3(\ub,u)}\lesssim&\|r^2\eta\cdot\etab+r^2|\etab|^2\|_{\L_u^1\H^3(\ub)}+\|r^2\rho\|_{\L_u^1\H^3(\ub)}\\
\lesssim&\varepsilon\Delta_0^2+\varepsilon^{1/2}\|r^2\rho\|_{\L_u^2\H^3(\ub)}\lesssim C(\mathcal{O}_0).
\end{align*}
We also apply Lemma \ref{L2LinftyH3} to obtain
\begin{align*}
\|r^{5/2}\omega\|_{\L^2_{[0,\ub_*]}\H^3(\ub,u)}\lesssim&\varepsilon\sup_{u}\|r^{5/2}\eta\cdot\etab+r^{5/2}|\etab|^2\|_{\L_{[0,\ub_*]}^2\H^3(u)}+\varepsilon\sup_{u}\|r^{5/2}\rho\|_{\L_{[0,\ub_*]}^2\H^3(u)}\\
\lesssim&\varepsilon\sup_{\ub,u}\|r^{3}\eta\cdot\etab+r^{3}|\etab|^2\|_{\H^3(\ub,u)}+\varepsilon\sup_{u}\|r^{5/2}\rho\|_{\L_{[0,\ub_*]}^2\H^3(u)}+\varepsilon\Delta_0^2\\
\lesssim&C(\mathcal{O}_0)
\end{align*}

We then consider the structure equations for $\Db(\Omega\tr\chi)$, $\Db\eta$ and $\Db\widetilde{\Omega\tr\chi}$, the last of which is derived by the equation for $\Db(\Omega\tr\chi)$. These equations involve the first order derivatives of the connection coefficients on the right hand side. Therefore, we can only estimate at most the second order derivatives of $\Omega\tr\chi$ and $\eta$ under the bootstrap assumption. By applying the Gronwall type estimates, if $\varepsilon$ is sufficiently small, we have
\begin{align*}\|r^2\widetilde{\Omega\tr\chi},r\Omega\tr\chi\|_{\H^2(\ub,u)}\lesssim& \mathcal{O}_0+\|r\eta\|_{\L_u^1\H^3(\ub)}+\|r^2|\eta|^2,r^2\Omega\tr\chib\Omega\tr\chi,r^2\Omega\chih\cdot\Omega\chibh\|_{\L_u^1\H^3(\ub)}\\
&+\|r^2\rho\|_{\L_u^1\H^2(\ub)}\\
\lesssim&C(\mathcal{O}_0)+\varepsilon\sup_{u}\|r\eta\|_{\H^3(\ub,u)}+\varepsilon\Delta_0^2+\varepsilon^{1/2}\|r^2\rho\|_{\L_u^2\H^2(\ub)}\\
\lesssim&C(\mathcal{O}_0)
\end{align*} 
and
\begin{align*}
\|r^{3/2}\eta\|_{\H^2(\ub,u)}\lesssim&\|r^{3/2}\eta\|_{\H^2(\ub,0)}+\|r^{1/2}\omegab\|_{\L_u^1\H^3(\ub)}+\|r^{3/2}\Omega\chib\cdot\eta\|_{\L_u^1\H^2(\ub)}+\|r^{3/2}\betab\|_{\L_u^1\H^2(\ub)}\\
\lesssim&C(\mathcal{O}_0,\mathcal{R}_0)+\varepsilon\Delta_0^2+\varepsilon^{1/2}\|r^2\betab\|_{\L_u^2\H^2(\ub,u)}\\
\lesssim&C(\mathcal{O}_0,\mathcal{R}_0).
\end{align*}
\begin{comment}
On the other hand, we also have
\begin{align*}
\|r^{3/2}\eta\|_{\L^2_{[0,\ub_*]}\H^2(u)}\lesssim&C(\mathcal{O}_0)+\varepsilon\sup_u\|r^{1/2}\omegab\|_{\L^2_{[0,\ub_*]}\H^3(u)}+\varepsilon\sup_{u}\|r^{3/2}\Omega\chib\cdot\eta\|_{\L_{[0,\ub_*]}^2\H^2(u)}\\
&+\varepsilon\sup_{u}\|r^{3/2}\betab\|_{\L_{[0,\ub_*]}^2\H^2(u)}\\
\lesssim&C(\mathcal{O}_0,\mathcal{R}_0)+\varepsilon\sup_{\ub,u}\left\{\|r\omegab\|_{\H^3(\ub,u)}+\|r^{2}\Omega\chib\cdot\eta\|_{\H^2(\ub,u)}\right\}\\&+\varepsilon\sup_{u}\|r^{3/2}\betab\|_{\L_{[0,\ub_*]}^2\H^2(u)}\\
\lesssim&C(\mathcal{O}_0,\mathcal{R}_0)
\end{align*}
if $\varepsilon$ is sufficiently small.
\end{comment}

We then use the structure equation for $\Db(\Omega\chih)$ to estimate
\begin{align*}
\|r^{3/2}\Omega\chih\|_{\H^2(\ub,u)}\lesssim&C(\mathcal{O}_0)+\varepsilon\sup_u\|r^{1/2}\eta\|_{\H^3(\ub,u)}
\\&+\varepsilon\sup_u\|r^{3/2}\eta\cdot\eta,r^{3/2}\Omega\tr\chi\Omega\chibh,r^{3/2}\Omega\tr\chib\Omega\chih\|_{\H^2(\ub,u)}\\\lesssim&C(\mathcal{O}_0)+\varepsilon\sup_{u}\|r^{1/2}\eta\|_{\H^3(\ub,u)}
+\varepsilon\Delta_0^2\\
\lesssim&C(\mathcal{O}_0)\end{align*}
if $\varepsilon$ is sufficiently small. We also estimate
\begin{align*}
\|r^{3/2}\Omega\chih\|_{\L^2_{[0,\ub_*]}\L^\infty_u\H^2(u)}\lesssim&C(\mathcal{O}_0)+\varepsilon\sup_u\|r^{1/2}\eta\|_{\L^2_{[0,\ub_*]}\H^3(u)}
\\&+\varepsilon\sup_u\|r^{3/2}\eta\cdot\eta,r^{3/2}\Omega\tr\chi\Omega\chibh,r^{3/2}\Omega\tr\chib\Omega\chih\|_{\L^2_{[0,\ub_*]}\H^2(u)}\\
\lesssim&C(\mathcal{O}_0)+\varepsilon\sup_{\ub,u}\|r\eta\|_{\H^3(\ub,u)}
+\varepsilon\sup_{\ub,u}\|r^2\eta\cdot\eta,r^2\Omega\tr\chi\Omega\chibh,r^2\Omega\tr\chib\Omega\chih\|_{\H^2(\ub,u)}\\
\lesssim&C(\mathcal{O}_0)
\end{align*}
if $\varepsilon$ is sufficiently small.

We then consider the structure equations along $\ub$ direction. First of all, from the above estimates, we can actually choose $\varepsilon$ sufficiently small such that $r^{3/2}|\chih|$, $r^2|\widetilde{\Omega\tr\chi}|$ are bounded by a constant $c$ depending on $\mathcal{O}_0$, therefore the first part of Lemma \ref{evolution} is true.

 The equations along $\ub$ direction should be integrated from the last slice $\Cb_{\ub_*}$. Therefore, we should first investigate the quantities on the last slice $\Cb_{\ub_*}$. Recall the equation which is satisfied by $\log\Omega$ on the last slice:
\begin{align*}
\Deltas\log\Omega=\frac{1}{2}\divs\etab+\frac{1}{2}\left(\frac{1}{2}((\chih,\chibh)-\overline{(\chih,\chibh)})-(\rho-\overline{\rho})\right),
\end{align*}
which is coupled with $\overline{\log\Omega}=0$. Because $\eta+\etab=2\ds\log\Omega$, we have
\begin{align}\label{lasteta}
\divs\eta=\frac{1}{2}((\chih,\chibh)-\overline{(\chih,\chibh)})-(\rho-\overline{\rho})
\end{align}
on the last slice. This equation should be viewed as the equation for $\eta$ on the last slice.

We can also deduce the equation satisfied by $\omegab$ on the last slice $\Cb_{\ub_*}$:
\begin{align}\label{lastomegab}
2\Deltas\omegab=2\divs(\Omega\betab)+\divs(3\Omega\chibh\cdot\eta+\frac{1}{2}\Omega\tr\chib\eta)+\Omega\tr\chib\divs\eta+(F-\overline{F}+\overline{\Omega\tr\chib\check\rho}-\overline{\Omega\tr\chib}\cdot\overline{\check{\rho}})
\end{align}
where $\check{\rho}=\rho-\frac{1}{2}(\chih,\chibh)$,
\begin{align*}
F=\frac{3}{2}\Omega\tr\chib\check{\rho}-(\ds\Omega,\betab)+\{(2\eta-\zeta,\betab)-\frac{1}{2}(\chibh,\nablas\tensor\eta+\eta\tensor\eta)+\frac{1}{4}\tr\chi|\chibh|^2\},
\end{align*}
and $\overline{\omegab}=-\overline{\Omega\tr\chib\log\Omega}$ by $\overline{\log\Omega}=0$.

To apply the elliptic estimates, we need the following lemma.
\begin{lemma}\label{curvatureonS} For $\varepsilon$ sufficiently small,
$$\|r^{5/2}\beta, r^{5/2}\rho, r^2\betab,r^2K\|_{\H^2(\ub,u)}, \|r^{5/2}\beta\|_{\L^2_{[0,\ub_*]}\L^\infty_u\H^2}\le C(\mathcal{R}_0,\mathcal{O}_0).$$
\end{lemma}
\begin{proof}
The proof relies on the null Bianchi equations for $\Db\betab$ and $\Db\rho$.
\begin{comment}
\begin{align*}
\Db\check{\rho}+\frac{3}{2}\Omega\tr\chib\check{\rho}&=\Omega\{-\divs\betab-(2\eta-\zeta,\betab)-\frac{1}{2}(\chibh,\nablas\tensor\eta+\eta\tensor\eta)+\frac{1}{4}\tr\chi|\chibh|^2\}.
\end{align*}
\end{comment}
By the Gronwall type estimates, if $\varepsilon$ is sufficiently small,
\begin{align*}
\|r^{5/2}\beta\|_{\H^2(\ub,u)}\lesssim&\|r^{5/2}\beta\|_{\H^2(\ub,0)}+\|r^{5/2}(\Omega\tr\chib,\Omega\chibh,\omegab)\cdot\beta,r^{5/2}\eta\cdot(\rho,\sigma),r^{5/2}\chih\cdot\betab\|_{\L_u^1\H^2(\ub)}\\
&+\|r^{3/2}\rho,r^{3/2}\sigma\|_{\L_u^1\H^3(\ub)},\\
\lesssim&\mathcal{R}_0+\varepsilon^{1/2}\Delta_0\mathcal{R}+\varepsilon^{1/2}\|r^2\rho,r^2\sigma\|_{\L_u^2\H^3(\ub)}\lesssim C(\mathcal{R}_0),
\end{align*}
\begin{align*}
\|r^{5/2}\rho\|_{\H^2(\ub,u)}\lesssim&\mathcal{R}_0+\|r^{5/2}\Omega\tr\chib\rho,r^{5/2}\eta\cdot\betab,r^{5/2}\etab\cdot\betab,r^{5/2}\chih\cdot\alphab\|_{\L_u^1\H^2(\ub)}+\|r^{3/2}\betab\|_{\L_u^1\H^3(\ub)}\\
\lesssim&\mathcal{R}_0+\varepsilon^{1/2}\Delta_0\mathcal{R}+\varepsilon^{1/2}\|r^2\betab\|_{\L_u^2\H^3(\ub)}\lesssim C(\mathcal{R}_0),
\end{align*}
\begin{align*}
\|r^2\betab\|_{\H^2(\ub,u)}\lesssim&\mathcal{R}_0+\|r^2(\Omega\tr\chib,\Omega\chibh,\omegab)\cdot\betab,r^2\eta\cdot\alphab,r^2\etab\cdot\alphab\|_{\L_u^1\H^2(\ub)}+\|r\alphab\|_{\L_u^1\H^3(\ub)}\\
\lesssim&\mathcal{R}_0+\varepsilon^{1/2}\Delta_0\mathcal{R}+\varepsilon^{1/2}\|r\alphab\|_{\L_u^2\H^3(\ub)}\lesssim C(\mathcal{R}_0).
\end{align*}
The estimate for $K$ comes from the Gauss equation $K+\frac{1}{4}\tr\chi\tr\chib-\frac{1}{2}(\chih,\chibh)=-\rho$.

We also estimate
\begin{align*}
&\|r^{5/2}\beta\|_{\L^2_{[0,\ub_*]}\L^\infty_u\H^2}\\
\lesssim&\mathcal{R}_0+\varepsilon\sup_{\ub,u}r\|\Omega\tr\chib,\Omega\chibh,\omegab,\eta\|_{\H^2(\ub,u)}\sup_u\|r^{3/2}\beta,r^{3/2}\rho,r^{3/2}\sigma,r^{3/2}\betab\|_{\L^2_{[0,\ub_*]}\H^2(u)}\\
&+\varepsilon\sup_u\|r^{3/2}\rho,r^{3/2}\sigma\|_{\L^2_{[0,\ub_*]}\H^3(u)}\\
\lesssim&\mathcal{R}_0+\varepsilon(\Delta_0+\mathcal{R})\mathcal{R}\lesssim C(\mathcal{R}_0)
\end{align*}
if $\varepsilon$ is sufficiently small.
\end{proof}

Now, we investigate $\Omega\tr\chi$ on the last slice. The equation for $\Db(\Omega\tr\chi)$ on the last slice reduces to
\begin{align*}
\Db(\Omega\tr\chi)+\frac{1}{2}\Omega\tr\chib\Omega\tr\chi&=\Omega^2(2|\eta|^2-2\overline{\check{\rho}}).
\end{align*}
We estimate
\begin{align*}
\|r^5\nablas^3(\Omega\tr\chi)\|_{\L^3(\ub_*,u)}\lesssim&\mathcal{O}_0+\|r^2\Omega\tr\chib\Omega\tr\chi,r^2|\eta|^2+r^2\overline{\check{\rho}}\Omega^2\|_{\L_u^1\H^3(\ub_*)}\\
\lesssim&\mathcal{O}_0+\varepsilon\Delta_1^2+\varepsilon\sup_u\left\{\|\check{\rho}\|_{\L^2(\ub_*,u)}\|\eta+\etab\|_{\H^2(\ub_*,u)}\right\}\lesssim C(\mathcal{O}_0,\mathcal{R}_0)
\end{align*}

We then consider the structure equation for $D(\nablas^3(\Omega\tr\chi))$ which is obtained by commuting $\nablas^3$ with the equation for $D(\Omega\tr\chi)$. Notice that the commutator $[D,\nablas^3]$ does not contain the third order derivatives of $\chih$ and $\tr\chi$, therefore we have
\begin{equation}\label{nablas3trchi}
\begin{split}
\|r^5\nablas^3(\Omega\tr\chi)\|_{\L^2(\ub,u)}\lesssim& C(\mathcal{O}_0,\mathcal{R}_0)+\int_{\ub}^{\ub_*}r^2\|r^3\nablas^3(\Omega\tr\chi)\|_{\L^2(\ub',u)}\|\widetilde{\Omega\tr\chi},\omega\|_{\H^2(\ub',u)}\D \ub'\\
&+\int_{\ub}^{\ub_*}r^2\|r^3\nablas^3(\Omega\chih)\|_{\L^2(\ub',u)}\|\Omega\chih,\widetilde{\Omega\tr\chi}\|_{\H^2(\ub',u)}\D\ub'\\
&+\int_{\ub}^{\ub_*}r^2\left[\|\Omega\chih\|_{\H^2(\ub',u)}^2+\|\widetilde{\Omega\tr\chi}\|_{\H^2(\ub',u)}\|\Omega\chih,\widetilde{\Omega\tr\chi},\omega\|_{\H^2(\ub',u)}\right]\D\ub'\\
&+\int_{\ub}^{\ub_*}r^2\|\Omega\tr\chi\|_{\L^2(\ub',u)}\|r^3\nablas^3\omega\|_{\L^2(\ub',u)}\D\ub'.
\end{split}
\end{equation}
We will firstly analyze the terms in the last two lines above. The third line is bounded by
\begin{align*}
\text{Line 3}\lesssim&\int_{\ub}^{\ub_*}r^2\left[\|\Omega\chih\|_{\H^2(\ub',u)}^2+\|\widetilde{\Omega\tr\chi}\|_{\H^2(\ub',u)}^2+\|\omega\|_{\H^2(\ub',u)}^2\right]\D\ub'\\
\lesssim&\|r^{3/2}\chih\|_{\L^2_{[\ub,\ub_*]}\H^2(u)}^2+\sup_{\ub}\|r^2\widetilde{\Omega\tr\chi},r^2\omega\|^2_{\H^2(\ub,u)}\lesssim C(\mathcal{O}_0).
\end{align*}
The fourth line is crucial because $\Omega\tr\chi$ itself will lose decay. Therefore the decay of $\nablas^3\omega$ should be good enough. By H\"older inequality, we estimate
\begin{align*}
\text{Line 5}\lesssim&\int_{\ub}^{\ub_*}r^{-3/2}\|r\Omega\tr\chi\|_{\L^2(\ub',u)}\|r^3\nablas^3(r^{5/2}\omega)\|_{\L^2(\ub',u)}\D\ub'\\
\lesssim& \left(\int_{\ub}^{\ub_*}r^{-3/2}C(\mathcal{O}_0)\D\ub'\right)^{1/2}\|r^3\nablas^3(r^{5/2}\omega)\|_{\L^2_{[\ub,\ub_*]}\L^2(u)}
\lesssim C(\mathcal{O}_0).
\end{align*}

We now go to the equation for $\divs(\Omega\chih)$. By elliptic estimate, we have
\begin{align}\label{ellipticchih}
\|r^3\nablas^3(\Omega\chih)\|_{\L^2(\ub,u)}\lesssim\|r^3\nablas^3(\Omega\tr\chi)\|_{\L^2(\ub,u)}+\|r\Omega\chih\cdot\etab,r\Omega\tr\chi\etab\|_{\H^2(\ub,u)}+\|r\beta\|_{\H^2(\ub,u)}
\end{align}
So, the second line is estimated as
\begin{align*}
\text{Line 2}\lesssim&\int_{\ub}^{\ub_*}\|r^5\nablas^3(\Omega\tr\chi)\|_{\L^2(\ub,u)}\|\Omega\chih,\widetilde{\Omega\tr\chi}\|_{\H^2(\ub',u)}\D\ub'\\
&+\int_{\ub}^{\ub_*}\left[r^{-3/2}\|r^{3/2}\Omega\chih\cdot r^{3/2}\etab\|_{\H^2(\ub'u)}+r^{-1}\|r\Omega\tr\chi r^{3/2}\etab\|_{\H^2(\ub,u)}+r^{-1}\|r^{5/2}\beta\|_{\H^2(\ub,u)}\right]\\
&\phantom{+\int_{\ub}^{\ub_*}}\times\|r^{3/2}\Omega\chih,r^{3/2}\widetilde{\Omega\tr\chi}\|_{\H^2(\ub',u)}\D\ub'\\
\lesssim&\int_{\ub}^{\ub_*}\|r^5\nablas^3(\Omega\tr\chi)\|_{\L^2(\ub,u)}\|\Omega\chih,\widetilde{\Omega\tr\chi}\|_{\H^2(\ub',u)}\D\ub'\\
&+\left[\sup_{\ub}\left\{\|r^{3/2}\|\Omega\chih\|_{\H^2(\ub,u)}\|r^{3/2}\etab\|_{\H^2(\ub,u)}+\|r\Omega\tr\chi\|_{\H^2(\ub,u)}\|r^{3/2}\etab\|_{\L^2_{[\ub,\ub_*]}\H^2(u)}\right\}\right.\\
&\phantom{+}\left.+\|r^{5/2}\beta\|_{\L^2_{[\ub,\ub_*]}\H^2(u)}\right]\times\left[\sup_{\ub}\|r^2\widetilde{\Omega\tr\chi}\|+\|r^{3/2}\Omega\chih\|_{\L^2_{[\ub,\ub_*]}\H^2(u)}\right]\\
\lesssim&\int_{\ub}^{\ub_*}\|r^5\nablas^3(\Omega\tr\chi)\|_{\L^2(\ub,u)}\|\Omega\chih,\widetilde{\Omega\tr\chi}\|_{\H^2(\ub',u)}\D\ub'+C(\mathcal{O}_0,\mathcal{R}_0).
\end{align*}
Notice that we use $\|r^{5/2}\beta\|_{\L^2_{[0,\ub_*]}\H^2(u)}\le\|r^{5/2}\beta\|_{\L^2_{[0,\ub_*]}\L^\infty_u\H^2}\le C(\mathcal{R}_0)$ by Lemma \ref{curvatureonS}.
Now, \eqref{nablas3trchi} becomes an integral inequality for $\|r^5\nablas^3(\Omega\tr\chi)\|_{\L^2(\ub,u)}$. To apply Gronwall inequality, we need the following:
\begin{align*}
\int_{\ub}^{\ub_*}\|\Omega\chih,\widetilde{\Omega\tr\chi},\omega\|_{\H^2(\ub',u)}\D \ub'\lesssim&\left(\int_{\ub}^{\ub_*}r^{-3/2}\D\ub'\right)^{1/2}\|r^{3/2}\Omega\chih,r^{3/2}\widetilde{\Omega\tr\chi},r^{3/2}\omega\|_{\L^2_{[\ub,\ub_*]}\H^2(\ub',u)}\\
\le& C(\mathcal{O}_0).
\end{align*}
Therefore, we have concluded that
\begin{align*}
\|r^5\nablas^3(\Omega\tr\chi)\|_{\L^2(\ub,u)}\lesssim C(\mathcal{O}_0,\mathcal{R}_0),
\end{align*}
and by Lemma \ref{curvatureonS}, \eqref{L2Linftyeta} and \eqref{ellipticchih}, we have
\begin{align*}
 \|r^{9/2}\nablas^3(\Omega\chih)\|_{\L^2_{[0,\ub_*]}\L^\infty_u\L^2(u)}\lesssim C(\mathcal{O}_0,\mathcal{R}),\\
 \|r^{9/2}\nablas^3\Omega\chih\|_{\L^2(u)}\lesssim C(\mathcal{O}_0,\mathcal{R}_0).
\end{align*}

Now we turn to $\eta$. We apply the elliptic estimate to \eqref{lasteta} on the last slice,
\begin{align*}
\|r^{3/2}\eta\|_{\H^3(\ub_*,u)}\lesssim\|r^{5/2}\rho\|_{\H^2(\ub_*,u)}+\|r^{5/2}\chih\cdot\chibh\|_{\H^2(\ub_*,u)}+\|r^{3/2}\eta\|_{\L^2(\ub_*,u)}\lesssim C(\mathcal{O}_0,\mathcal{R}_0).
\end{align*}
We then apply the Gronwall type estimates to the equation for $D\nablas^3\eta$. Notice that the commutator $[D,\nablas^3]\eta=\sum_{i+j=2}\nablas^{i+1}(\Omega\chi)\nablas^j\eta$ does not contain $\nablas^3\eta$, and then does not contain new terms that are not estimated before. Therefore, we have
\begin{align*}
&\|r^4\nablas^3\eta\|_{\L^2(\ub,u)}\\
\lesssim& \|r^4\nablas^3\eta\|_{\L^2(\ub_*,u)}+\int_{\ub}^{\ub_*}r\left[\|\Omega\chi\|_{\H^3(\ub',u)}(\|\eta\|_{\H^2(\ub',u)}+\|\etab\|_{\H^3(\ub',u)})+\|\beta\|_{\H^3(\ub',u)}\right]\D\ub'\\
\lesssim& \|r^4\nablas^3\eta\|_{\L^2(\ub_*,u)}+\int_{\ub}^{\ub_*}r^{-3/2}\left[\|r\Omega\chi\|_{\H^3(\ub',u)}(\|r^{3/2}\eta\|_{\H^2(\ub',u)}+\|r^{3/2}\etab\|_{\H^3(\ub',u)})\right.\\
&\left.\phantom{ \|r^4\nablas^3\eta\|_{\L^2(\ub_*,u)}+\int_{\ub}^{\ub_*}r^{-3/2}}+\|r^{5/2}\beta\|_{\H^3(\ub',u)}\right]\D\ub'\\
\lesssim& \|r^4\nablas^3\eta\|_{\L^2(\ub_*,u)}+r^{-1/2}\sup_{\ub}\left\{\|r\Omega\chi\|_{\H^3(\ub,u)}(\|r^{3/2}\eta\|_{\H^2(\ub,u)}+\|r^{3/2}\etab\|_{\H^3(\ub,u)})\right\}\\
&\phantom{\|r^4\nablas^3\eta\|_{\L^2(\ub_*,u)}}+r^{-1/2}\|r^{5/2}\beta\|_{\L^2_{[\ub,\ub_*]}\H^3(u)}.
\end{align*}
Therefore, multiplying both sides by $r^{1/2}$, we have
\begin{align*}
\|r^{9/2}\nablas^3\eta\|_{\L^2(\ub,u)}\le C(\mathcal{O}_0,\mathcal{R}_0,\mathcal{R}).
\end{align*}

Finally we turn to $\omega$. By elliptic estimate (Lemma \ref{curvatureonS} and \ref{elliptic}), we have
\begin{align*}
\|r\omegab\|_{\H^3(\ub_*,u)}\lesssim \|r^3\Deltas\omegab\|_{\H^1(\ub_*,u)}+\|r\omegab\|_{\L^2(\ub_*,u)}\lesssim C(\mathcal{O}_0,\mathcal{R}_0).
\end{align*}

We then apply the Gronwall type estimates for $D\nablas^i\omegab, i=0,1,2,3$, we have
\begin{align*}
\|\omegab\|_{\H^3(\ub,u)}\lesssim& \|\omegab\|_{\H^3(\ub_*,u)}+\int_{\ub}^{\ub_*}\left[\|\Omega\chih,\widetilde{\Omega\tr\chi}\|_{\H^2(\ub',u)}\|\omegab\|_{\H^2(\ub'u)}\right.\\
&\left.\phantom{ \|\omegab\|_{\H^3(\ub_*,u)}+\int_{\ub}^{\ub_*}\ }+\|2(\eta,\etab)-|\eta|^2\|_{\H^3(\ub',u)}+\|\rho\|_{\H^3(\ub',u)}\right]\D\ub'.
\end{align*}
By the Gronwall inequality, the first term in the integral is absorbed and we then have
\begin{align*}
\|\omegab\|_{\H^3(\ub,u)}\lesssim& \|\omegab\|_{\H^3(\ub_*,u)}+r^{-2}\sup_{\ub}\{\|r^{3/2}\eta\|_{\H^3(\ub,u)}^2+\|r^{3/2}\eta\|_{\H^3(\ub,u)}\|r^{3/2}\etab\|_{\H^3(\ub,u)}\}\\
&+r^{-3/2}\|r^{5/2}\rho\|_{\L^2_{[\ub,\ub_*]}\H^3(u)}.
\end{align*}
Therefore, multiplying both sides by $r$, we have
\begin{align*}
\|r\omegab\|_{\H^3(\ub,u)}\lesssim C(\mathcal{O}_0,\mathcal{R}_0,\mathcal{R}).
\end{align*}

\subsection{Proof of Proposition \ref{curvaturecompletenullcone}}\label{SectionCurvature}

The proof relies on the following energy estimates.
\begin{lemma}\label{energyestimate}
There exist functions $\tau_j^{(i)}$ with $i=0,1,2$ and $j=0,1,2,3$ such that we have, for every $\ub\in[0,\ub_*]$,
\begin{align*}
\|r^{5/2}\alpha\|_{\L^2_{[0,\ub]}\H^3(u)}^2+\|r^2\beta\|_{\L^2_u\H^3(\ub)}^2\lesssim& \|r^{5/2}\alpha\|_{\L^2_{[0,\ub_*]}\H^3(0)}^2+\|r^2\beta\|_{\L^2_u\H^3(0)}^2\\
&+\int_0^\varepsilon\int_0^{\ub}\sum_{i=0}^3r^{4+2i}\|\tau_0^{(i)}\|_{\L^1(\ub',u')}\D\ub'\D u',\\
\|r^{5/2}\beta\|_{\L^2_{[0,\ub]}\H^3(u)}^2+\|r^2\rho,r^2\sigma\|_{\L^2_u\H^3(\ub)}^2\lesssim& \|r^{5/2}\beta\|_{\L^2_{[0,\ub_*]}\H^3(0)}^2+\|r^2\rho,r^2\sigma\|_{\L^2_u\H^3(0)}^2\\
&+\int_0^\varepsilon\int_0^{\ub}\sum_{i=0}^3r^{4+2i}\|\tau_1^{(i)}\|_{\L^1(\ub',u')}\D\ub'\D u',\\
\|r^{5/2}\rho,r^{5/2}\sigma\|_{\L^2_{[0,\ub]}\H^3(u)}^2+\|r^2\betab\|_{\L^2_u\H^3(\ub)}^2\lesssim& \|r^{5/2}\rho,r^{5/2}\sigma\|_{\L^2_{[0,\ub_*]}\H^3(0)}^2+\|r^2\betab\|_{\L^2_u\H^3(0)}^2\\
&+\int_0^\varepsilon\int_0^{\ub}\sum_{i=0}^3r^{4+2i}\|\tau_2^{(i)}\|_{\L^1(\ub',u')}\D\ub'\D u',\\
\|r^{3/2}\beta\|_{\L^2_{[0,\ub]}\H^3(u)}^2+\|r\alphab\|_{\L^2_u\H^3(\ub)}^2\lesssim& \|r^{3/2}\beta\|_{\L^2_{[0,\ub_*]}\H^3(0)}^2+\|r^2\alphab\|_{\L^2_u\H^3(0)}^2\\
&+\int_0^\varepsilon\int_0^{\ub}\sum_{i=0}^3r^{2+2i}\|\tau_3^{(i)}\|_{\L^1(\ub',u')}\D\ub'\D u',\\
\end{align*}
with $\tau_j^{(i)}$ satisfying the following estimates:
\begin{align}
\begin{split}\label{tau0e}
\sum_{i=1}^3\|\tau_0^{(i)}\|_{\L^1(\ub,u)}\lesssim&\|\Omega\tr\chib,\chibh,\omegab\|\|\alpha\|^2+\|\eta,\etab\|\|\alpha\|\|\beta\|+\|K\|_{\H^2(\ub,u)}\|\alpha\|\|\beta\|+\|\chih\|\|\alpha\|\|\rho,\sigma\|\\&+\|\Omega\tr\chi,\chih,\omega\|\|\beta\|^2,
\end{split}\\
\begin{split}\label{tau1e}
\sum_{i=1}^3\|\tau_1^{(i)}\|_{\L^1(\ub,u)}\lesssim&\|\Omega\tr\chib,\chibh,\omegab\|\|\beta\|^2+\|\eta,\etab\|\|\beta\|\|\rho,\sigma\|+\|\chibh\|\|\alpha\|\|\rho,\sigma\|+\|K\|_{\H^2(\ub,u)}\|\beta\|\|\rho,\sigma\|\\&+\|\chih\|\|\beta\|\|\betab\|+\|\Omega\tr\chi,\chih\|\|\rho,\sigma\|^2,
\end{split}\\
\begin{split}\label{tau2e}
\sum_{i=1}^3\|\tau_2^{(i)}\|_{\L^1(\ub,u)}\lesssim&\|\Omega\tr\chib,\chibh\|\|\rho,\sigma\|^2+\|\eta,\etab\|\|\rho,\sigma\|\|\betab\|+\|\chibh\|\|\beta\|\|\betab\|+\|K\|_{\H^2(\ub,u)}\|\rho,\sigma\|\|\betab\|\\&+\|\chih\|\|\rho,\sigma\|\|\alphab\|+\|\widetilde{\Omega\tr\chi},\chih,\omega\|\|\betab\|^2,
\end{split}\\
\begin{split}\label{tau3e}
\sum_{i=1}^3\|\tau_3^{(i)}\|_{\L^1(\ub,u)}\lesssim&\|\Omega\tr\chib,\chibh,\omegab\|\|\betab\|^2+\|\eta,\etab\|\|\betab\|\|\alphab\|+\|\chibh\|\|\rho,\sigma\|\|\alphab\|+\|K\|_{\H^2(\ub,u)}\|\betab\|\|\alphab\|\\&+\|\widetilde{\Omega\tr\chi},\chih,\omega\|\|\alphab\|^2,
\end{split}
\end{align}
where the norm $\|\cdot\|$ above refers to $\|\cdot\|_{\H^3(\ub,u)}$.
\end{lemma}

\begin{proof}
We can simply define the error terms $\tau_0^{(i)}$, $\tau_1^{(i)}$, $\tau_2^{(i)}$, $\tau_3^{(i)}$ to be
\begin{align}\label{tau0}
r^{2+2i}\tau_0^{(i)}\D\mu_{\gs}=&\Db(r^{2+2i}|\nablas^i\alpha|^2\D\mu_{\gs})+D(2r^{2+2i}|\nablas^i\beta|^2\D\mu_{\gs})-r^{2+2i}\divs(4\Omega\nablas^i\alpha\cdot\nablas^i\beta)\D\mu_{\gs},\\
\begin{split}\label{tau1}
r^{2+2i}\tau_1^{(i)}\D\mu_{\gs}=&\Db(r^{2+2i}|\nablas^i\beta|^2\D\mu_{\gs})+D(r^{2+2i}(|\nablas^i\rho|^2+|\nablas^i\sigma|^2)\D\mu_{\gs})\\
&-2r^{2+2i}\divs(\Omega(\nablas^i\rho\cdot\nablas^i\beta-\nablas^i\sigma\cdot\nablas^i{}^*\beta))\D\mu_{\gs},
\end{split}\\
\begin{split}\label{tau2}
r^{2+2i}\tau_2^{(i)}\D\mu_{\gs}=&\Db(r^{2+2i}(|\nablas^i\rho|^2+|\nablas^i\sigma|^2))\D\mu_{\gs})+D(r^{2+2i}|\nablas^i\betab|^2\D\mu_{\gs})\\
&+2r^{2+2i}\divs(\Omega(\nablas^i\rho\cdot\nablas^i\beta-\nablas^i\sigma\cdot\nablas^i{}^*\beta))\D\mu_{\gs},
\end{split}\\
\label{tau3}
r^{2i}\tau_3^{(i)}\D\mu_{\gs}=&\Db(2r^{2i}|\nablas^i\betab|^2\D\mu_{\gs})+D(r^{2i}|\nablas^i\alphab|^2\D\mu_{\gs})+\divs(4\Omega\nablas^i\alphab\cdot\nablas^i\betab)\D\mu_{\gs}.
\end{align}
The estimates \eqref{tau0e}-\eqref{tau3e} then can be derived by commuting $\nablas^i$ with the null Bianchi equations.  Notice that the Gauss curvature $K$ comes from the commutator of $\nablas$ and the Hodge operators. We omit the details here. Then the estimates \eqref{tau0e}-\eqref{tau3e} follows by integrating \eqref{tau0}, \eqref{tau1}, \eqref{tau2}, \eqref{tau3} over $\displaystyle\bigcup_{(\ub',u)\in[0,\ub]\times[0,\varepsilon]} S_{\ub',u}$ for all $\ub\in[0,\ub_*]$.
\end{proof}

One should notice that in the estimates of $\tau_2^{(i)}$ and $\tau_3^{(i)}$, $\widetilde{\Omega\tr\chi}$ appears instead of $\Omega\tr\chi$.

By \eqref{tau0e}, \eqref{tau1e}, we have
\begin{align*}
&\int_0^\varepsilon\int_0^{\ub}\sum_{i=0}^3\sum_{k=0}^1r^{4+2i}\|\tau_k^{(i)}\|_{\L^1(\ub',u')}\D\ub'\D u'\\
\lesssim &\sup_{\ub,u}\|\Omega\tr\chib,\chibh,\omegab\|_{\H^3(\ub,u)}\int_0^{\varepsilon}\|r^{5/2}\alpha,r^{5/2}\beta\|^2_{\L^2_{[0,\ub]}\H^3(u')}\D u'\\
&+\sup_{\ub,u}r(\|\eta,\etab,\chibh,\chih\|_{\H^3(\ub,u)}+\|K\|_{\H^2(\ub,u)})\left(\int_0^{\varepsilon}\|r^{5/2}\alpha,r^{5/2}\beta\|^2_{\L^2_{[0,\ub]}\H^3(u')}\D u'\right)^{1/2}\\
&\phantom{+}\times\left(\int_0^{\varepsilon}\|r^{3/2}\beta,r^{3/2}\rho,r^{3/2}\sigma\|^2_{\L^2_{[0,\ub]}\H^3(u')}\D u'\right)^{1/2}\\
&+\sup_{\ub,u}\|\Omega\tr\chi,\chih,\omega\|_{\H^3(\ub,u)}\int_0^{\varepsilon}\|r^{5/2}\beta,r^{5/2}\rho,r^{5/2}\sigma\|^2_{\L^2_{[0,\ub]}\H^3(u')}\D u'\\
\lesssim& \varepsilon\mathcal{O}\mathcal{R}^2
\end{align*}

By  \eqref{tau2e}, we have
\begin{align*}
&\int_0^\varepsilon\int_0^{\ub}\sum_{i=0}^3r^{4+2i}\|\tau_2^{(i)}\|_{\L^1(\ub',u')}\D\ub'\D u'\\
\lesssim &\sup_{\ub,u}\|\Omega\tr\chib\|_{\H^3(\ub,u)}\int_0^{\varepsilon}\|r^{5/2}\rho,r^{5/2}\sigma\|^2_{\L^2_{[0,\ub]}\H^3(u')}\D u'\\
&+\sup_{\ub,u}r(\|\eta,\etab,\chibh\|_{\H^3(\ub,u)}+\|K\|_{\H^2(\ub,u)})\left(\int_0^{\varepsilon}\|r^{5/2}\rho,r^{5/2}\sigma\|^2_{\L^2_{[0,\ub]}\H^3(u')}\D u'\right)^{1/2}\\
&\phantom{+}\times\left(\int_0^{\varepsilon}\|r^{3/2}\betab\|^2_{\L^2_{[0,\ub]}\H^3(u')}\D u'\right)^{1/2}\\
&+\underline{\|r^{3/2}\chih\|_{\L^2_{[0,\ub]}\L^\infty_u\H^3(u)}\|r^{5/2}\rho,r^{5/2}\sigma\|_{\L^2_{[0,\ub]}\H^3(u)}\cdot\varepsilon^{1/2}\|r\alphab\|_{\L^2_u\H^3(\ub)}}\\
&+\sup_{\ub,u}r^{3/2}\|\widetilde{\Omega\tr\chi},\chih,\omega\|_{\H^3(\ub,u)}\int_0^{\ub}(1+\ub')^{-3/2}\|r^2\betab\|^2_{\L^2_{u}\H^3(\ub')}\D \ub'\\
\lesssim& \varepsilon^{1/2}\mathcal{O}\mathcal{R}^2+\mathcal{O}\int_0^{\ub}(1+\ub')^{-3/2}\|r^2\betab\|^2_{\L^2_{u}\H^3(\ub')}\D \ub'
\end{align*}
The factor $(1+\ub)^{-3/2}$ in the last term above, which ensures the convergence,  is exactly the reason why we need $\widetilde{\Omega\tr\chi}$ appears instead of $\Omega\tr\chi$ in the estimate of $\tau_2^{(i)}$. This is also the case in the estimate of $\tau_3^{(i)}$. The term with underline which comes from the estimate of $\tau_2^{(i)}$ is the borderline term. And this is exactly the place the norm $\L^2_{[0,\ub]}\L^\infty_u\H^3$ comes in.

By \eqref{tau3e}, we have
\begin{align*}
&\int_0^\varepsilon\int_0^{\ub}\sum_{i=0}^3r^{2+2i}\|\tau_3^{(i)}\|_{\L^1(\ub',u')}\D\ub'\D u'\\
\lesssim &\sup_{\ub,u}\|\Omega\tr\chib,\chibh,\omegab\|_{\H^3(\ub,u)}\int_0^{\varepsilon}\|r^{3/2}\betab\|^2_{\L^2_{[0,\ub]}\H^3(u')}\D u'\\
&+\sup_{\ub,u}r(\|\eta,\etab,\chibh\|_{\H^3(\ub,u)}+\|K\|_{\H^2(\ub,u)})\left(\int_0^{\varepsilon}\|r^{3/2}\rho,r^{3/2}\sigma,r^{3/2}\betab\|^2_{\L^2_{[0,\ub]}\H^3(u')}\D u'\right)^{1/2}\\
&\phantom{+}\times\left(\int_0^{\ub}r^{-3/2}\|r\alphab\|^2_{\L^2_u\H^3(\ub')}\D \ub'\right)^{1/2}\\
&+\sup_{\ub,u}r^{3/2}\|\widetilde{\Omega\tr\chi},\chih,\omega\|_{\H^3(\ub,u)}\int_0^{\ub}(1+\ub')^{-3/2}\|r\alphab\|^2_{\L^2_{u}\H^3(\ub')}\D \ub'\\
\lesssim&\varepsilon^{1/2}\mathcal{O}\mathcal{R}^2+\mathcal{O}\int_0^{\ub}(1+\ub')^{-3/2}\|r\alphab\|^2_{\L^2_{u}\H^3(\ub')}\D \ub'
\end{align*}

By the first two estimates in Lemma \ref{energyestimate} and the above estimates for $\tau_0^{(i)}$ and $\tau_1^{(i)}$, we have
\begin{align*}
&\|r^{5/2}\alpha,r^{5/2}\beta\|^2_{\L^2_{\ub}\H^3(u)}+\|r^2\beta,r^2\rho,r^2\sigma\|^2_{\L^2_{u}\H^3(\ub)}\lesssim\mathcal{R}_0^2+\varepsilon\mathcal{O}\mathcal{R}^2.
\end{align*}
By Proposition \ref{connection}, $\mathcal{O}\le C(\mathcal{O}_0,\mathcal{R}_0,\mathcal{R})$, therefore if $\varepsilon$ is chosen sufficiently small, we have
\begin{align*}
\|r^{5/2}\alpha,r^{5/2}\beta\|^2_{\L^2_{\ub}\H^3(u)}+\|r^2\beta,r^2\rho,r^2\sigma\|^2_{\L^2_{u}\H^3(\ub)}\lesssim C(\mathcal{R}_0).
\end{align*}
\begin{remark}
The differences between the initial norms $\|r^2\beta,r^2\rho,r^2\sigma,r^2\betab,r\alphab\|_{\L^2_u\H^3(0)}^2$ and the corresponding norms in $\mathcal{R}_0$ are determined by all third order derivatives (including $\nablas$, $\Db$ and their mixed derivatives) of $\Omega$, which are bounded by a constant depending only on $\mathcal{O}_0$, $\mathcal{R}_0$ if $\varepsilon$ is sufficiently small.
\end{remark}

By the last two estimates in Lemma \ref{energyestimate} and the above estimates for $\tau_2^{(i)}$ and $\tau_3^{(i)}$,
\begin{align*}
\|r^2\betab\|^2_{\L^2_{u}\H^3(\ub)}+\|r\alphab\|^2_{\L^2_{u}\H^3(\ub)}\lesssim& \mathcal{R}_0+\varepsilon^{1/2}\mathcal{O}\mathcal{R}^2\\
&+\mathcal{O}\int_0^{\ub}(1+\ub')^{-3/2}(\|r^2\betab\|^2_{\L^2_{u}\H^3(\ub')}+\|r\alphab\|^2_{\L^2_{u}\H^3(\ub')})\D\ub'.
\end{align*}
Notice that $\mathcal{O}$ appearing in the third term on the right hand side comes from the estimates for $\chih$, $\widetilde{\Omega\tr\chi}$ and $\omega$, which depend only on $\mathcal{O}_0$, $\mathcal{R}_0$. Therefore, by choosing $\varepsilon$ sufficiently small, and absorbing the last term, we have
\begin{align*}
\|r^2\betab\|^2_{\L^2_{u}\H^3(\ub)}+\|r\alphab\|^2_{\L^2_{u}\L^\infty_{[0,\ub]}\H^3}\lesssim C(\mathcal{O}_0,\mathcal{R}_0).
\end{align*}
We then substitute this estimate back to the estimates for $\tau_2^{(i)}$ and $\tau_3^{(i)}$, and again by the last two estimates in Lemma \ref{energyestimate}, we have
\begin{align*}
\|r^{5/2}\rho,r^{5/2}\sigma,r^{3/2}\betab\|^2_{\L^2_{\ub}\H^3(u)}\lesssim C(\mathcal{O}_0,\mathcal{R}_0).
\end{align*}
Therefore, we complete the proof of Proposition \ref{curvaturecompletenullcone} and Theorem \ref{apriori}.

\section{Appendix}

As mentioned above, we will sketch the proof of Step 1 and Step 3, which is the construction of the canonical foliation on an incoming null cone. We refer the readers to \cite{L-Z} and the references therein for the details. 

Recall that a canonical foliation on an incoming null cone $\Cb_{\ub}$ means that, under this foliation, the following equation holds:
\begin{align*}
\overline{\log\Omega}=0,\quad \Deltas\log\Omega=\frac{1}{2}\divs\etab+\frac{1}{2}\left(\frac{1}{2}((\chih,\chibh)-\overline{(\chih,\chibh)})-(\rho-\overline{\rho})\right),
\end{align*}
which is exactly \eqref{lastslice}. Now we work on an arbitrary incoming null cone $\Cb_{\ub}$ with a background foliation given by a function $s$. If we work on the initial slice $\Cb_0$, the background foliation refers to that given by the affine function $\underline{s}$, and if we work on the last slice $\Cb_{\ub_*+\delta}$, the background foliation refers to that given by $u$, which is obtained by extending the outgoing null cones $C_u$ to $\Cb_{\ub_*+\delta}$, and gives a canonical foliation on $\Cb_{\ub_*}$,

We use a function $W=W(s,\theta)$ defined on $[0,\varepsilon]\times S_{\ub,0}$ to represent a new foliation in the following way: The new foliation function ${}^{(W)}u$  is defined by the relation ${}^{(W)}u(W(s,\theta),\theta)=s$. Under the new foliation, we have also the new ``lapse'' function ${}^{(W)}\Omega$. $W$ represents a foliation iff
\begin{align*}
{}^{(W)}a(s,\theta)\triangleq\frac{\partial W}{\partial s}(s,\theta)={}^{(W)}\Omega^2(W(s,\theta),\theta)\Omega^{-2}(W(s,\theta),\theta)>0.
\end{align*}
Here $\Omega$ is the lapse function relative to the background foliation given by $s$. In general, $W$ does not necessarily represent a foliation but only a family of spherical sections parameterized by $s$.

Now given $W\ge0$. We consider a map $\mathcal{A}$, such that $\mathcal{A}(W)\ge0$ is again a family of spherical sections, which is defined by the following:
\begin{align}\label{definitionA}
\mathcal{A}(W)(s,\theta)=\int_0^s{}^{(W_\mathcal{A})}\Omega^2(W(s',\theta),\theta)\Omega^{-2}(W(s',\theta),\theta)\D s'
\end{align}
 where ${}^{(W_\mathcal{A})}\Omega$ is the solution of the equation
\begin{align*}
{}^{(W)}\Deltas\log{}^{(W_\mathcal{A})}\Omega(W(s,\theta),\theta)&={}^{(W)}G(W(s,\theta),\theta),\\
{}^{{}^{(W)}}\overline{\log{}^{(\mathcal{A}(W))}\Omega}(s)&=0,
\end{align*}
where
\begin{align*}
{}^{(W)}G\triangleq\frac{1}{2}{}^{(W)}\divs{}^{(W)}\etab+\frac{1}{2}\left(\frac{1}{2}(({}^{(W)}\chih,{}^{(W)}\chibh)-{}^{{}^{(W)}}\overline{({}^{(W)}\chih,{}^{(W)}\chibh)})-({}^{(W)}\rho-{}^{{}^{(W)}}\overline{{}^{(W)}\rho})\right).
\end{align*}
If $W$ is a fixed point of $\mathcal{A}$, then the foliation given by $W$ is canonical. The strategy is to find a suitable closed subspace of the functions $W$, and prove that $\mathcal{A}$ restricted to this subspace is a contraction.

 We introduce the following norms on $\Cb_{\ub}$ with respect to the background foliation given by $s$:
\begin{align*}
\mathbb{O}(\ub)&=\sup_s\|r\chibh,r\tr\chib,r^{3/2}\eta,r^{3/2}\etab,r\omegab\|_{\H^2(\ub,s)}+\sup_s\|r\Db\omegab\|_{\H^1(\ub,s)}+\sup_s\|r\Db^2\omegab\|_{\L^2(\ub,s)}\\
\underline{\mathbb{R}}(\ub)&=\sup_s\|r^{5/2}\rho,r^{5/2}\sigma,r^2\betab,r\alphab\|_{\H^2(\ub,s)}+\sup_s\|r\Db\alphab\|_{\H^1(\ub,s)}+\sup_s\|r\Db^2\alphab\|_{\L^2(\ub,s)}.
\end{align*}
Using the argument in \cite{L-Z}, we can prove the following
\begin{proposition}Given two families of sections represented by $W_1,W_2$. If $\mathbb{O}(\ub),\underline{\mathbb{R}}(\ub)\le C$, then if $\varepsilon$, $W_i$ and $\sup_s\|r^{1/2}(r\nablas)^{1,2}W_i(s,\cdot)\|_{\L^2(\ub,s)}, i=1,2$ are sufficiently small (depending on the initial data $\mathcal{O}_0$, $\mathcal{R}_0$ and $C$), there exists a constant $c$ depending on $C$ such that
\begin{equation}\label{contraction}
\begin{split}
\|r^{1/2}(\log{}^{((W_1)_{\mathcal{A}})}\Omega(W_1(s,\cdot),\cdot)-\log{}^{((W_2)_{\mathcal{A}})}\Omega(W_2(s,\cdot),\cdot))\|_{\H^2(\ub,0)}\\
\le c\|r^{1/2}(W_1(s,\cdot)-W_2(s,\cdot))\|_{\H^2(\ub,0)}.
\end{split}
\end{equation}
\end{proposition}
\begin{remark}
Notice that the regularity of the norms in $\mathbb{O}$ and $\underline{\mathbb{R}}$ are stronger than those used in \cite{L-Z}. This is because we use three order derivatives of the curvature in this paper, therefore the second order derivatives of the curvature components are controlled. On the other hand, the function $W-s$, and the lapse $\Omega-1$, decay only like $r^{-1/2}$ (but not $r^{-1}$) because $\rho$ decays only like $r^{-5/2}$ (but not $r^{-3}$).
\end{remark}

The construction of the canonical foliation on $\Cb_0$ and $\Cb_{\ub_*+\delta}$ is then a direct consequence. We discuss the more subtle case for $\Cb_{\ub_*+\delta}$ and the case for $\Cb_0$ is similar. We introduce the closed subspace $\mathcal{K}=\mathcal{K}_{\ub_*+\delta,\varepsilon+\delta'}\subset C([0,\varepsilon],H^2(S_{\ub_*+\delta,0}))$ such that
\begin{equation}\label{functionspacelast}
\begin{cases}W(0,\theta)=0,\\0\le W(s,\theta)\le\varepsilon+\delta',\\\sup_s\|r^{1/2}(W(s,\cdot)-s)\|_{\H^2(\ub_*+\delta,0)}\le\varepsilon_\mathcal{K},\end{cases}
\end{equation}
for some small number $\varepsilon_\mathcal{K} $ to be fixed. Recall that $s$ here is actually the origin optical function $u$, which induces a canonical foliation on $\Cb_{\delta}$. The second condition ensures that the foliation given by $W\in\mathcal{K}$ does not go beyond $s\in[0,\varepsilon+\delta']$. Notice that we can choose $\varepsilon_\mathcal{K}$ small enough such that the second condition holds. Then the existence of a function $u_\delta$ that induces a canonical foliation on $\Cb_{\ub_*+\delta}$ follows from the following proposition:
\begin{proposition}\label{fixedpoint}
For $\varepsilon$ sufficiently small depending on $\mathcal{O}_0$, $\mathcal{R}_0$, and $\delta$ sufficiently small (may depending on $\varepsilon$ and $\delta'$), $\mathcal{A}(\mathcal{K})\subset\mathcal{K}$ and $\mathcal{A}$ is a contraction in $\mathcal{K}\subset C([0,\varepsilon],H^2(S_{\ub_*+\delta,0}))$.
\end{proposition}

The proof of the above proposition is similar as in \cite{L-Z}, using \eqref{definitionA} and \eqref{contraction}, provided that $\mathbb{O}(\ub_*+\delta),\underline{\mathbb{R}}(\ub_*+\delta)\le C(\mathcal{O}_0,\mathcal{R}_0)$. The estimates for $\mathbb{O},\underline{\mathbb{R}}$ are already done, except $\sigma$ and its angular derivatives, $\Db\omegab$ and its angular derivative, $\Db^2\omegab$, $\Db\alphab$ and its angular derivative, and $\Db^2\alphab$. The estimates for $\sigma$ are obtained in a similar way to $\rho$, as in Lemma \ref{curvatureonS}. The estimates for $\Db\alphab$ and its angular derivative can be obtained by commuting $\Db$ and $\nablas\Db$ with the null Bianchi equation for $D\alphab$. After commuting derivatives, by taking into account the null structure equations and Bianchi equations, the right hand side contains at most third order angular derivatives of the connection coefficients and curvature components. Then we can apply the Gronwall type estimates. The estimate for $\Db^2\alphab$ can then be done in a similar way, by commuting $\Db^2$ with the equation for $D\alphab$. The estimates for $\Db\omegab$ and its angular derivatives, and $\Db^2\omegab$, rely on the equation \eqref{lastomegab}, which is written on $\Cb_{\ub_*}$. We commute $\Db$ and $\Db^2$ with this equation, and use elliptic estimate. We can then choose $\delta$ sufficiently small, to conclude that the estimates on $\Cb_{\ub_*+\delta}$ still hold.

Finally, using the argument as in \cite{L-Z}, if $\varepsilon$ is chosen sufficiently small, we can extend $u_\delta$ back to the whole spacetime, in the way that the level sets of $u_\delta$ are outgoing null cones, that are orthogonal to the spherical level sets of $u_\delta$ on $\Cb_{\ub_*+\delta}$.

\end{document}